\documentclass[
    pra,
    onecolumn,
    nofootinbib,
    superscriptaddress
]{revtex4-2}
\setcitestyle{authoryear,round}
\newlength{\bibitemsep}
\newlength{\bibparskip}
\let\oldthebibliography\thebibliography
\renewcommand\thebibliography[1]{%
  \oldthebibliography{#1}%
  \setlength{\parskip}{\bibitemsep}%
  \setlength{\itemsep}{\bibparskip}%
}

\usepackage{physics}
\usepackage{nicefrac}
\usepackage{tikz}
\usetikzlibrary{positioning,shapes.geometric,decorations.pathreplacing,arrows.meta}
\usepackage{bm}
\usepackage{float}
\usepackage{mathrsfs}
\usepackage[linktocpage]{hyperref}
\usepackage[amsthm,amsmath,thmmarks,hyperref]{ntheorem}
\usepackage[normalem]{ulem}

\newtheorem{theorem}{Theorem}

\newtheorem{lemma}{Lemma}
\newtheorem{corollary}{Corollary}[theorem]

\theoremstyle{definition}
\newtheorem{optimization}{Optimization}
\theoremstyle{remark}
\newtheorem{remark}{Remark}[theorem]
\theoremstyle{definition}
\newtheorem{definition}{Definition}
\newtheorem{example}{Example}[definition]

\usepackage[caption=false,justification=raggedright,singlelinecheck=false]{subfig}

\usepackage{algorithm}
\usepackage[noend]{algpseudocode}
\usepackage{amssymb}
\usepackage{graphicx}
\usepackage{enumitem}
\usepackage{xcolor}
\usepackage{epigraph} 
\hypersetup{
    colorlinks,
    linkcolor={blue!50!black},
    citecolor={blue!50!black},
    urlcolor={blue!80!black}
}

\usepackage[capitalize,noabbrev,nameinlink]{cleveref}
\Crefname{optimization}{Optimization}{Optimizations}

\DeclareMathOperator*{\polylog}{polylog}
\newcommand\numberthis{\addtocounter{equation}{1}\tag{\theequation}}
\DeclareMathOperator\supp{supp}

\usepackage{rotating}

\newcommand{\stkout}[1]{\ifmmode\text{\sout{\ensuremath{#1}}}\else\sout{#1}\fi}
\newif\ifverbose
\verbosetrue

\begin{document}
\tikzset{meter/.append style={draw, inner sep=10, rectangle, font=\vphantom{A}, minimum width=30, line width=.8,
 path picture={\draw[black] ([shift={(.1,.3)}]path picture bounding box.south west) to[bend left=50] ([shift={(-.1,.3)}]path picture bounding box.south east);\draw[black,-latex] ([shift={(0,.1)}]path picture bounding box.south) -- ([shift={(.3,-.1)}]path picture bounding box.north);}}}

\title{Practical Black Box Hamiltonian Learning}

\author{Andi Gu}
\email{andi.gu@berkeley.edu}
\affiliation{Department of Physics, University of California, Berkeley, Berkeley, CA 94720, USA}
\affiliation{Theoretical Division, Los Alamos National Laboratory, Los Alamos, NM 87545, USA}



\author{Lukasz Cincio}
\affiliation{Theoretical Division, Los Alamos National Laboratory, Los Alamos, NM 87545, USA}

\author{Patrick J. Coles}
\affiliation{Theoretical Division, Los Alamos National Laboratory, Los Alamos, NM 87545, USA}

\date{\today}

\begin{abstract}
We study the problem of learning the parameters for the Hamiltonian of a quantum many-body system, given limited access to the system. 
In this work, we build upon recent approaches to Hamiltonian learning via derivative estimation. We propose a protocol that improves the scaling dependence of prior works, particularly with respect to parameters relating to the structure of the Hamiltonian (e.g., its locality $k$). Furthermore, by deriving exact bounds on the performance of our protocol, we are able to provide a precise numerical prescription for theoretically optimal settings of hyperparameters in our learning protocol, such as the maximum evolution time (when learning with unitary dynamics) or minimum temperature (when learning with Gibbs states). Thanks to these improvements, our protocol is practical for large problems: we demonstrate this with a numerical simulation of our protocol on an 80-qubit system. 
\end{abstract}

\maketitle

\section{Introduction}


An important task for learning about many-body quantum systems is to learn the associated Hamiltonian operator efficiently (i.e., without requiring resources that scale exponentially in system size). There are multiple physical motivations for this problem. First, for an isolated many-body system, all the information about its future state is contained in the initial state and the Hamiltonian $H$, since the Hamiltonian generates the time evolution operator $e^{-iHt}$. Second, for a system in thermal contact with an environment of inverse temperature $\beta$, the thermal equilibrium state is of the form $e^{-\beta H} / \Tr(e^{-\beta H})$, and hence is also determined by the Hamiltonian. Third, the spectral response, such as the absorption and emission spectra, of a system is determined by its Hamiltonian.

The Hamiltonian learning problem is a highly relevant task across many domains. In condensed matter physics, we can experimentally verify our models of quantum materials by comparing theoretical predictions about their effective interactions with the interactions inferred by Hamiltonian learning \citep{burgarth2017, wang2017, kwon2020, wang2020}. This verification is also applicable for quantum device engineering. With the expanding capabilities of quantum computers, it is increasingly important to be able to certify their behavior \citep{carrasco2021}. While benchmarking protocols can give coarse-grained information about a particular quantum device, knowing its Hamiltonian can be significantly more powerful, allowing us to design improved devices \citep{boulant2003,innocenti2020,ben2020} or better understand the physical origin of failure modes \citep{shulman2014,sheldon2016,sundaresan2020}. 

In this work, we will treat the system under study as a black box system with an unknown Hamiltonian $H$, and our goal will be to efficiently infer $H$ with access to only a limited number of inputs to, and outputs from the black box. Importantly, we assume that we can only interact with the system \textit{classically} (see \cref{fig:classic-inter}). The key defining characteristic of `classical' interaction is that we prohibit any quantum channel between the system under study (whose Hamiltonian we are trying to learn), and some other quantum processing unit -- that is, we rule out setups used in other approaches that assume we can interact with the system under study via another trusted quantum simulator \citep[e.g.,][]{wiebe2014a,wiebe2014b,verdon2019}. Two examples of what we call `classically limited' interactions are making measurements on time-evolved states (\cref{subfig:time-ev}) or on Gibbs states (\cref{subfig:gibbs}). For the former, we initialize the system in some known (mixed) state $\rho_0$, and evolve it forward in time by $t$, resulting in the state:
\begin{equation}
    \rho(t) = e^{-iHt} \rho_0 e^{iHt}  \label{eqn:time-state}.
\end{equation}
For the latter, we assume we have access to a system in thermal equilibrium at a temperature $\beta^{-1}$. That is, we have access to the Gibbs state
\begin{equation}
    \rho(\beta) = \frac{\exp(-\beta H)}{\Tr(\exp(-\beta H))} \label{eqn:gibbs-state}.
\end{equation}
In these two models of interaction, we assume we can control the parameters $t$ and $\beta$, respectively. Finally, we assume that we can measure some observable $O$ of the final states $\rho(t)$ and $\rho(\beta)$. However, we do not demand full control over $\rho_0$ and $O$, since arbitrary quantum states are hard to prepare \citep{plesch2011}. We only require that $\rho_0$ and $O$ be a tensor product over single sites. More precisely, we assume that we can prepare any $\rho_0 \in \qty{I/2, (I+\sigma_x)/2, (I+\sigma_y)/2, (I+\sigma_z)/2}^{\otimes n}$ and measure any observable $O \in \qty{I, \sigma_x, \sigma_y, \sigma_z}^{\otimes n}$. Our focus on these two methods of interacting with the system is in part motivated by the classical analogue of the quantum Hamiltonian learning problem. This classical analogue is well-studied by the machine learning community, with two primary approaches being learning using system dynamics \citep{brunton2016,pan2016,trischler2016,lusch2018,course2021} or samples from the Gibbs distribution \citep{abbeel2006,santhanam2012,bresler2013,lokhov2018}. 

\begin{figure}
    \centering
    \subfloat[\textbf{Time evolution}: we can control three quantities: $\rho_0$, $t$, and $O$. We assume we can evolve the input state $\rho_0$ forward in time. After a time $t$, we make a measurement of the observable $O$.]{
    \label{subfig:time-ev}
    \parbox{.4\textwidth}{
    \begin{tikzpicture}
    \node (time) [rectangle, minimum width=2.2cm, minimum height=1.2cm,text centered, draw=black] {$e^{-iH t}$};
    \draw[-] (time.165) --++(0:-0.5cm) node[above,yshift=0.1cm] (top) {};
    \draw[-] (time.173) --++(0:-0.5cm);
    \draw[-] (time.west) --++(0:-0.5cm);
    \draw[-] (time.187) --++(0:-0.5cm);
    \draw[-] (time.195) --++(0:-0.5cm) node[below,yshift=-0.1cm] (bot) {};
    
    \draw[-] (time.15) --++(0:0.5cm);
    \draw[-] (time.7) --++(0:0.5cm);
    \draw[-] (time.353) --++(0:0.5cm);
    \draw[-] (time.345) --++ (0:0.5cm);
    
    \draw[dashed, thick] (-2.7,-1.45) rectangle (3,1.2);
    \draw[-] (time.0) --++(0:0.5cm) node[right, meter] (met) {};
    \draw[double,-{Implies},double distance=0.3mm] (met.south) --++ (270:0.55cm) -- ++(180:2.2cm) --++ (270:0.8cm) node (out) {};
    \draw[double,{Implies}-,double distance=0.3mm] (time.north) --++(90:0.9cm) node[above] {$t$};
    \node[yshift=-0.4cm] at (out) {$f(t)$};
    \node[yshift=0.8cm] at (met) {$O$};

    \draw [decorate, decoration = {brace, raise=5pt},thick] (bot) --  (top) node[pos=0.5,left=10pt]{$\rho_0$};
    \end{tikzpicture}
    }}
    \hspace{1cm}
    \subfloat[\textbf{Gibbs states}: we can control two quantities: $\beta$ and $O$. We assume we have access to the Gibbs state at a temperature $\beta^{-1}$, and then measure the observable $O$.]{
    \label{subfig:gibbs}
    \parbox{.38\textwidth}{
    \begin{tikzpicture}
    \node (time) [rectangle, minimum width=2.5cm, minimum height=1.2cm,text centered, draw=black] {$\frac{\exp(-\beta H)}{\Tr(\exp(-\beta H))}$};
    
    \draw[-] (time.15) --++(0:0.5cm);
    \draw[-] (time.7) --++(0:0.5cm);
    \draw[-] (time.0) --++(0:0.5cm) node[right, meter] (met) {};
    \draw[-] (time.353) --++(0:0.5cm);
    \draw[-] (time.345) --++(0:0.5cm);
    
    \draw[dashed, thick] (-1.6,-1.45) rectangle (3,1.2);
    \draw[-] (time.0) --++(0:0.5cm) node[right, meter] (met) {};
    \draw[double,-{Implies},double distance=0.3mm] (met.south) --++ (270:0.55cm) -- ++(180:2.2cm) --++ (270:0.8cm) node (out) {};
    \draw[double,{Implies}-,double distance=0.3mm] (time.north) --++(90:0.9cm) node[above] {$\beta$};
    \node[yshift=-0.4cm] at (out) {$f(\beta)$};
    \node[yshift=0.8cm] at (met) {$O$};
    \end{tikzpicture}}
    
    }
    \caption[Classical interaction with quantum systems]{Two models of classical interaction with a black box quantum system. We view the system as a set of oracles indexed by the state preparation and measurement parameters $\rho_0, O$ in the time evolution case, and $O$ in the Gibbs state case. These oracles take some input $t$ or $\beta$, and we use their output to characterize the Hamiltonian.}
    \label{fig:classic-inter}
\end{figure}
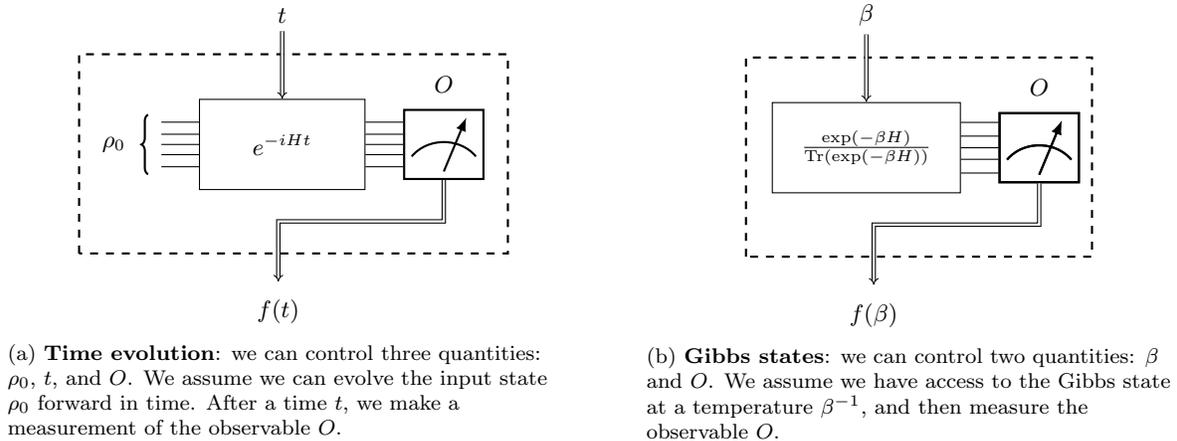

Using these two models of black-box interaction, we propose a method for Hamiltonian learning that relies on a simple intuition. For some state preparation and measurement (SPAM) settings, we can define a function $f_{\text{SPAM}}$ as the expectation value of an observable (which is specified by the SPAM settings) on the final state $\rho(t)$ and $\rho(\beta)$:
\begin{equation}
    f_\text{SPAM}(x) = \begin{cases}
    \Tr(O \rho(t=x)) &\qq{for unitary evolution} \\
    \Tr(O \rho(\beta=x)) &\qq{for Gibbs states.}
    \end{cases} \label{eqn:polynomial-pic}
\end{equation}
We will show that for the appropriate choice of SPAM parameters, $f_{\text{SPAM}}(x)$ can be viewed as black box function in $x$, whose Taylor expansion can be connected in a straightforward manner to the coefficients of the Hamiltonian. More specifically, for each SPAM setting, the first order coefficient $f'(x=0)$ will yield one of the Hamiltonian's coefficients. In this work, we will formalize this idea into a protocol that is efficient in theory, and demonstrate its usefulness in practice.

\section{Preliminaries}

Before describing the contributions of this work, we give a formalized definition of the Hamiltonian learning problem, and define the class of Hamiltonians that we restrict our attention to.
\begin{definition}[Hamiltonian learning problem]\label{defn:ham-learn}
Fix a Hamiltonian on an $n$-qubit system that has an expansion in the Pauli basis:
\begin{equation}
    H = \sum_{m=1}^r \theta_m P_m ,
\end{equation}
where each $P_m \in \qty{I, \sigma_x, \sigma_y, \sigma_z}^{\otimes n}$ is a Pauli operator and $\theta_m \in \mathbb{R}$ are the Hamiltonian coefficients. We will denote the vector of Hamiltonian coefficients $\Theta = \mqty[\theta_1, \ldots, \theta_r]^T$. We assume the Hamiltonian is traceless (i.e., $P_m \neq I^{\otimes n}$), and that we know the structure of the Hamiltonian (i.e., which Paulis $P_m$ are present in the expansion), but that the coefficients $\theta_m$ are unknown. The Hamiltonian learning problem is to infer all of the coefficients $\theta_m$ up to an additive error $\epsilon \cdot \max_m \abs{\theta_m}$ with success probability at least $1-\delta$. 
\end{definition}

In this work, we restrict our attention to a broad class of Hamiltonians that we call \textit{sparsely interacting} Hamiltonians. A sparsely interacting Hamiltonian is defined below.
\begin{definition}[Sparsely interacting Hamiltonian]\label{defn:inter-graph}
The interaction graph (called the ``dual'' interaction graph in \citet{tang2021}) of a Hamiltonian $\mathcal{G}$ consists of a set of vertices $V$ and edges $E$.
\begin{gather}
    V = \qty{P_i \mid i=1,\ldots,r} \ , \\
    E = \qty{(P_i, P_j) \mid \qty(\supp(P_i) \cap \supp(P_j) \neq \varnothing) \land (i \neq j)} \ .
\end{gather}
Each vertex represents one Pauli operator $P_i$ in the Hamiltonian, and there are edges between two vertices if the support of their corresponding Pauli operators overlap. The support of a Pauli, $\supp (P)$, is the set of sites that $P$ acts nontrivially on. We also define the degree of the Hamiltonian $\mathscr{D}$ to be the maximum degree of any node in the interaction graph:
\begin{equation}
    \mathscr{D} = \max_{v \in V} \text{deg}(v) \ .
\end{equation}
A Hamiltonian is sparsely interacting if $\mathscr{D} = \order{1}$ (that is, $\mathscr{D}$ does not depend on system size). Notably, this class of Hamiltonians includes $k$-local Hamiltonians, as this locality constraint implies that the number of terms overlapping with any Pauli term is a function of $k$ alone.

\end{definition}
\begin{example}\label{ex:tfim-model}
Below, we show a sample interaction graph for a 9-qubit transverse field Ising model (TFIM), whose Hamiltonian is
\begin{equation}
    H = \sum_{i=1}^8 \sigma_z^{(i)} \sigma_z^{(i+1)} + \sum_{i=1}^9 \sigma_{x}^{(i)}.
\end{equation}
The TFIM will serve as a prototypical example for the rest of this work.
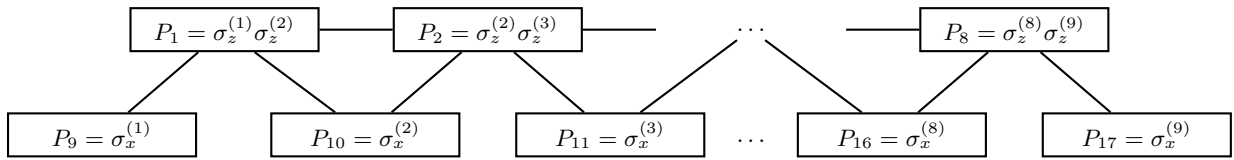
\begin{figure}[H]
    \centering
    \begin{tikzpicture}[node distance={35mm}, thick, main/.style = {draw, rectangle,minimum width=2.5cm}
    ] 
\node[main] (0) {$P_1=\sigma_z^{(1)} \sigma_z^{(2)}$}; 
\node[main] (1) [right of=0] {$P_2=\sigma_z^{(2)} \sigma_z^{(3)}$};
\node (2) [right of=1, minimum width=2.5cm, rectangle] {$\ldots$};
\node[main] (3) [right of=2] {$P_8=\sigma_z^{(8)} \sigma_z^{(9)}$};
\node[main] (4) [below left=8mm and -9mm of 0] {$P_9=\sigma_x^{(1)}$};
\node[main] (5) [below left=8mm and -9mm of 1] {$P_{10}=\sigma_x^{(2)}$};
\node[main] (6) [below right=8mm and -9mm of 1] {$P_{11}=\sigma_x^{(3)}$};
\node (7) [below=12mm of 2, minimum width=2.5cm, rectangle] {$\ldots$};
\node[main] (8) [below left=8mm and -9mm of 3] {$P_{16}=\sigma_x^{(8)}$};
\node[main] (9) [below right=8mm and -9mm of 3] {$P_{17}=\sigma_x^{(9)}$};

\draw[-] (0) -- (1);
\draw[-] (1) -- (2);
\draw[-] (2) -- (3);
\draw[-] (0) -- (4);
\draw[-] (0) -- (5);
\draw[-] (1) -- (5);
\draw[-] (1) -- (6);
\draw[-] (2) -- (6);
\draw[-] (2) -- (8);
\draw[-] (3) -- (8);
\draw[-] (3) -- (9);

\end{tikzpicture} 
    \caption[Interaction graph for a transverse field Ising model]{Interaction graph $\mathcal{G}$ for a $9$-qubit transverse field Ising model. The degree of this Hamiltonian is $\mathscr{D}=4$, since for instance $P_2$ is connected to 4 other Pauli terms.}
    \label{fig:int-graph}
\end{figure}
\end{example}

\section{Prior work and our contribution}

The assumption of a sparsely interacting (or $k$-local) Hamiltonian can often simplify the learning problem. For instance, in an early work \citep{dasilva2011}, it was shown that systems with local Hamiltonians can be efficiently characterized without the expensive requirements of full state tomography. However, this method was only applicable in certain cases, and was found to be prohibitively expensive in general. Later approaches demonstrated that machine learning could be applied successfully on small systems \citep{hentschel2010,hentschel2011,sergeevich2011,granade2012}, but these methods did not come with rigorous performance guarantees or scaling results that would allow them to be applied on larger systems. There have also been a number of proposals \citep{qi2019,bairey2019,evans2019} to learn the coefficients of the Hamiltonian by solving a system of linear equations, where the coefficient matrix is determined by local measurement outcomes. However, the performance of these approaches is determined by the spectral gap of this coefficient matrix, which is not yet well-characterized. 

A recent series of works \citep{anshu2021,tang2021,sbahi2022} focus on Hamiltonian learning with Gibbs states. Notably, \citet{tang2021} has found asymptotically optimal (with respect to error and failure probability) methods for Hamiltonian learning from high-temperature Gibbs states. They propose an algorithm that requires a number of copies of the Gibbs state (at temperature $\beta^{-1}$) that scales with $\beta^{-2}$ and polynomially in $\mathscr{D}$ (see \cref{defn:inter-graph}). However, they impose the high-temperature constraint $\beta^{-1} \geq 25e^6 (\mathscr{D}+1)^{10}$: Although this still results in polynomial scaling with $\mathscr{D}$, this is not feasible in practice, as this constraint on $\beta$ results in a query complexity with a prefactor that scales with $\mathscr{D}^{21}$. For one of the simplest non-trivial cases where we might apply Hamiltonian learning, the 1-dimensional transverse field Ising model (which has $\mathscr{D}=4$), the required number of copies of the Gibbs state has a prefactor $\gtrsim \mathscr{D} \beta^{-2} \geq 10^{22}$. Despite this, the techniques developed in this work are useful -- particularly, the central role of $\mathscr{D}$ in calculations.

Finally, there have also been proposals to learn the Hamiltonian from short time unitary dynamics \citep{hangleiter2021,yu2022,franca2022}. \citet{yu2022} considers the case where the Hamiltonian structure is not known \textit{a~priori}, and develop a protocol that is robust to circuit noise and SPAM errors, and avoids the exponential query requirements of full tomography. However, their protocol has relatively strong requirements on how to interact with the Hamiltonian: it requires the periodic insertion of gates between different time evolutions. Furthermore, their analysis does not include the effects of evolution time $t_0$ on their algorithm's performance. \citet{franca2022} has developed a protocol that has similar scaling to \citet{tang2021}. Notably, it can incorporate Lindbladian terms, as well as being compatible with algebraically decaying interactions. However, the protocol does not take advantage of prior knowledge about Pauli terms present in the Hamiltonian, which can potentially result in a measurement parallelization overhead that scales with $16^k$ (where $k$ is the locality of the Hamiltonian). 

In \cref{chap:methods}, we will describe an algorithm that addresses some of the shortcomings of previous works. Our protocol has a query complexity that scales like that achieved by \citet{tang2021}, except our dependence on the parameter $\mathscr{D}$ will be $\order{\mathscr{D}^4}$ rather than $\order{\mathscr{D}^{21}}$. Furthermore, we parallelize our measurements in such a way that avoids the $\order{16^k}$ scaling of \citet{franca2022}, and instead requires a parallelization overhead that scales with $\order{\mathscr{D}^2}$. This offers an improvement when we have prior information about the Pauli terms present in the Hamiltonian, so that $\mathscr{D} < 4^k$; indeed, a more typical assumption is $\mathscr{D} \sim \order{\text{poly}(k)}$. In summary, we achieve a query complexity
\begin{equation}
    \order{\epsilon^{-2} \mathscr{D}^4 \log (r/\delta) \polylog(\mathscr{D}/\epsilon)},
\end{equation}
and classical processing time complexity
\begin{equation}
    \order{\epsilon^{-2} \mathscr{D}^2 r \log (r/\delta) \polylog(\mathscr{D}/\epsilon)}.
\end{equation}
for Hamiltonian learning using unitary dynamics. Similar to \citet{franca2022}, this can be generalized, via careful selection of initial states and measurements, to learn the Lindbladian (when expanded in the Pauli basis) of open quantum systems undergoing Markovian dynamics. The query and classical processing time complexity using Gibbs states is only worse by a factor $\mathscr{D}$ and $\mathscr{D}^2$, respectively.  Finally, in \cref{chap:results}, we will discuss heuristic optimizations of our algorithm that can improve performance in practice. We will then show numerical results on an 80-qubit transverse field Ising model that indicates the feasibility our algorithm in practice.

\section{Methods}\label{chap:methods}

\subsection{Outline of our approach}

We outline our basic approach below. The following can be considered a generalization of the ideas in \citet{franca2022}. This generalization allows us to learn using Gibbs states as well as unitary dynamics, and is also able to reduce measurement parallelization overhead. 
\begin{enumerate}
    \item \textbf{Connect the oracle and the Hamiltonian parameters} \ Write the Taylor expansion of \cref{eqn:polynomial-pic}:
    \begin{equation}
        f_{\text{SPAM}}(x) = \sum_{k=0}^\infty c_k \frac{x^k}{k!}.
    \end{equation}
    Determine the connection between the first-order coefficient $c_1$ and the Hamiltonian. For both unitary dynamics and Gibbs states, this connection can be made extremely straightforward by appropriately setting the SPAM parameters.
    \item \textbf{Bound higher order derivatives} Establish a bound for the higher order derivatives $\abs{c_k}$ in terms of the structure parameter $\mathscr{D}$. This bound is similar in spirit to the shadow norm in shadow tomography \citep{huang2020}. The scaling we find for $\abs{c_k}$ varies depending on whether we are using unitary dynamics or Gibbs states (just as $\norm{O}_{\text{shadow}}^2$ may vary depending on whether random Clifford or Pauli measurements are used). Furthermore, similarly to \citep{huang2020}, this bound determines the performance of the rest of our algorithm. In this work, we find
    \begin{equation}
        \abs{c_k} \sim \begin{cases}
        \order{\mathscr{D}^k k!} &\qq{for general Hamiltonians using unitary dynamics} \\
        \order{\mathscr{D}^k} &\qq{for commuting Hamiltonians using unitary dynamics} \\
        \order{\mathscr{D}^{2k} k!} &\qq{for general Hamiltonians with Gibbs states.} \label{eqn:bounds} \\
        \end{cases}
    \end{equation}
    
    \item \textbf{Recover Hamiltonian parameters} \ Evaluate $f_{\text{SPAM}}$ at $L$ different points $x \in [0,A]$. Using a special form of polynomial regression, we can guarantee that $c_1$ (hence the Hamiltonian parameters) can be estimated with an error $\order{\frac{A^L \abs{c_L}}{L!}}$. Observe that if $\abs{c_L}$ grows no faster than a factorial, as is the case in \cref{eqn:bounds}, the error decreases (at least) as a power law in $L$ for suitably chosen $A$. However, our overall error scaling is nowhere near as good as this due to the presence of noise when evaluating $f_{\text{SPAM}}$ (increasing $L$ will result in an increase in the variance of our estimator for $c_1$). The modeling error (bias) must be carefully traded against the effects of noise (variance).
    \item \textbf{Apply simultaneous measurements} \ If possible, devise a scheme to execute the previous step using simultaneous measurements to estimate several parameters at once. When the SPAM settings use measurements whose locality does not exceed that of the Hamiltonian (as is the case in this work), this is possible. This enables a query complexity that is sublinear in the number of Hamiltonian coefficients $r$. 
\end{enumerate}
In \cref{sec:infer-first}, we first establish an elementary procedure for estimating the first order derivative $f'(0)$ given access only to noisy estimates of $f$. Then, in \cref{sec:recover}, we apply this procedure to Hamiltonian learning with unitary dynamics and Gibbs states.


\subsection{Inferring the First-Order Commutator}\label{sec:infer-first}

For a system evolving under a Hamiltonian $H$ and an initial state given by some density matrix $\rho_0$, the expectation value of any operator $P$ can be written as:
\begin{gather}
    \expval{P(t)} = \Tr(P \rho_0(t)) =\Tr(P e^{-iHt} \rho_0 e^{iHt}) = \sum_{m=0}^\infty \frac{(it)^m}{m!} \Tr(\qty[H^m P] \rho_0) \label{eqn:heis}, \\
    \textrm{where }\qty[H^m P] = \underbrace{[H,[H,\ldots,[H}_{m \text{ times}},P]\ldots]] \qq{with} [H^0 P] = P.
\end{gather}
This equality is simply using the Heisenberg expansion of the time-evolved operator $P(t)$. 

In this section, we define a critical subroutine of our Hamiltonian learning algorithm that infers the expectation $\Tr(\qty(i \comm{H}{P}) \rho_0)$ by measuring time-evolved expectation values. The main idea behind our algorithm is that $\Tr(\qty(i \comm{H}{P}) \rho_0)$ is the time derivative of the expectation $\Tr(P e^{-iH t} \rho_0 e^{i H t})$. More specifically, the Heisenberg expansion in \cref{eqn:heis} expresses the time-evolved expectation of an observable as
\begin{equation}
    \expval{P(t)} = \sum_{m=0}^\infty \frac{i^m}{m!} \Tr(\qty[H^m P] \rho_0) t^m .
\end{equation}
Therefore $\expval{P(t)}$ can be modeled as a univariate power series in time, $\sum_{m=0}^\infty c_m t^m$, with coefficients
\begin{equation}
    c_m = \frac{i^m}{m!} \Tr(\qty[H^m P] \rho_0) . \label{eqn:taylor-coeff}
\end{equation}
If we were able to access $\expval{P(t)}$ exactly, the most effective way to find $c_1$ would be to simply differentiate $\expval{P(t)}$ via finite differences with very small $\Delta t$ (i.e., $c_1 \approx \frac{\expval{P(\Delta t)}-\expval{P(0)}}{\Delta t}$). Since our measurements of $\expval{P(\Delta t)}$ are subject to shot noise, the variance of this estimator scales with $\order{(\Delta t)^{-2}}$, preventing us from using arbitrarily small $\Delta t$. However, as $\Delta t$ grows, the bias in the finite difference estimator grows. \cref{alg:first-comm} is a generalization of finite differencing, and uses Chebyshev regression (see \cref{sec:cheb}) to estimate $c_1$. This algorithm takes as input a maximum evolution time $A$ and an cutoff degree for the Chebyshev polynomial $L$. This finite cutoff degree induces biases in the recovered polynomial coefficients\footnote{There are biases in both the Chebyshev and Taylor expansion bases, hence the notation $\tilde{b}_i, \tilde{c}_i$ as opposed to the true coefficients $b_i$ and $c_i$.}, however, we will demonstrate that this bias is suppressed much more effectively than for the finite-difference estimator, as it turns out that these errors scale in a power-law with power $L$. As mentioned in the beginning of this section, this error bound depends on a bound for the derivative $\abs{\dv[L]{\expval{P(t)}}{t}} = \abs{\Tr(\qty[H^m P] \rho(t))}$. Since $\rho(t)$ is a density matrix, a simple application of the Von Neumann trace inequality \citep{mirsky1975} shows that $\abs{\Tr(\qty[H^m P] \rho(t))} \leq \norm{\qty[H^m P]}$ (where $\norm{\cdot}$ denotes the spectral norm). We can bound spectral norms of iterated commutators with the Hamiltonian as follows:

\begin{theorem}[Iterated commutator norm bound]\label{lem:norm-bound}
When $P$ is a single-qubit observable, the spectral norm of the iterated commutator is bounded by
\begin{equation}
    \norm{\qty[H^m P]} \leq (2 \mathscr{D} \norm{\Theta}_\infty)^m (m+1)! .\label{eqn:time-deriv-bound}
\end{equation}
where the $\ell_\infty$ norm denotes $\norm{\Theta}_\infty = \max_i \abs{\theta_i}$. For commuting Hamiltonians (i.e., every term $P_i$ in the Hamiltonian commutes with every other term), 
\begin{equation}
    \norm{\qty[H^m P]} \leq (2(\mathscr{D}+1)\norm{\Theta}_\infty)^m .\label{eqn:commute-bound}
\end{equation}
\end{theorem}
\begin{proof}
See \cref{sec:bound-iter}.
\end{proof}

\begin{definition}[Typical scales]\label{defn:typical-scale}
The form of the above bound is indicative of two typical scales. We will take
\begin{equation}
    \tau = \frac{1}{2\mathscr{D} \norm{\Theta}_\infty}
\end{equation}
to define a typical time scale for our Hamiltonian, and
\begin{equation}
    \gamma = \tau^{-1} = 2 \mathscr{D} \norm{\Theta}_\infty
\end{equation}
to define a typical scale for our Hamiltonian coefficients.
\end{definition}

\begin{algorithm}[H]
\caption{Estimating the first derivative $\Tr(\qty(i \comm{H}{P}) \rho_0)$}\label{alg:first-comm}
\begin{algorithmic}[1]
\Procedure {EstimateDerivative}{$A$, $L$, $P$, $N$, $\rho_0$}
\For{$\ell \gets 1, L$}   \Comment{Construct the dataset $\mathcal{D}$ (\cref{defn:dataset})}
\State $z_\ell \gets -\cos(\frac{2\ell-1}{L} \pi)$
\State $t_\ell \gets \frac{A}{2}(1+z_\ell)$ 
\State $y_\ell \gets \text{estimate of }\Tr(P e^{-iH t_{\ell}} \rho_0 e^{iH t_{\ell}})$ \Comment{Average $N$ measurement outcomes of $P$}
\EndFor
\For{$m \gets 1, L-1$} \Comment{Estimate the Chebyshev coefficients (\cref{thm:cheb-fit})}
\State $\tilde{b}_m \gets \frac{2}{L}\sum_{\ell=1}^{L} y_\ell T_m(z_\ell)$ 
\EndFor
\State{$\tilde{c}_1 \gets -\frac{2}{A} \sum_{m=1}^{L-1} (-1)^m \tilde{b}_m m^2$} \Comment{Infer $\Tr(i \comm{H}{P} \rho_0)$ (\cref{lem:linear})}
\State \textbf{return} $\tilde{c}_1$
\EndProcedure
\end{algorithmic}
\end{algorithm}

\begin{definition}[Dataset]\label{defn:dataset}
We construct a dataset with $L$ evaluations of the expectation $\expval{P(t)}$, evaluated at the roots of the $L$th Chebyshev polynomial. Our dataset comprises of $L$ points:
\begin{equation} \label{eqn:dataset}
\begin{gathered}
    \mathcal{D} = \qty{(t_1, y_1), (t_2, y_2), \ldots, (t_L, y_L)} \text{, where} \\
    t_i = \frac{A}{2} (1+z_i), \\
    y_i \sim Y_i,
\end{gathered}
\end{equation}
with $z_i$ the roots of the $L$th Chebyshev polynomial and $Y_i$ is a random variable with $\mathbb{E}[Y_i]=\expval{P(t_i)}$ and $\mathbb{V}[Y_i]=\sigma_i^2$. The notation $y_i \sim Y_i$ indicates that $y_i$ is a random sample of $Y_i$. The mapping $t_i = \frac{A}{2} (1+z_i)$ ensures that the evolution time is nonnegative and never exceeds $A$. 
\end{definition}

The following theorem shows that for the appropriate choice of evolution time $A$ and Chebyshev degree $L$, the error scaling is close to being noise-limited. 

\begin{theorem}[Query complexity for one coefficient]\label{thm:err-scale}
Fix some failure probability $\delta$ and an error $\epsilon$. Assume that we have access to an unbiased (single-shot) estimator of $\expval{P(t)}$ with variance $\sigma^2 \leq 1$. Then there is some choice of $A \sim \tau$ and $L \sim \log \epsilon^{-1}$ such that with
\begin{equation}
    \order{\log(1/\delta) \polylog(1/\epsilon) \epsilon^{-2}}
\end{equation}
query complexity, we can construct an estimator $\tilde{c}_1$ such that $\frac{\abs{c_1-\tilde{c}_1}}{\gamma} \leq \epsilon$, except with failure probability at most $\delta$.
\end{theorem}
\begin{proof}
See \cref{sec:proof-err}.
\end{proof}

\subsection{Recovering Hamiltonian Coefficients}\label{sec:recover}
With an efficient algorithm for accurately estimating first-order commutators $\Tr(i \comm{H}{P} \rho_0)$, it is possible to construct an algorithm that can infer the coefficients of $H$ using these commutators. The idea is to carefully choose $\rho_0$ and $P$ so that $\Tr(i \comm{H}{P} \rho_0)$ corresponds to one parameter at a time.

First, we introduce the notation that $\rho_0^{(X)}$ and $P^{(X)}$ will be the portion of a density matrix or Pauli matrix (respectively) that is restricted to the qubits in $X$, and $X'$ will be the set of all qubits not in $X$.

\begin{lemma}[Term selection]\label{thm:param-select}
Let $P$ be some Pauli operator such that there exists some $i \in \qty{1, \ldots, r}$ where $\supp P \subseteq \supp P_i$ and $\frac{i\comm{P_i}{P}}{2} \neq 0$. Let
\begin{gather}
    X = \supp P_i ,\\
    Y = \qty(\bigcup \qty{\supp P_j \mid \supp P_j \cap X \neq \varnothing}) \setminus X, \\
    Z = (X \cup Y)', \\
    \rho_0 = \qty(\frac{\mathbb{I} + i\comm{P_i}{P}/2}{2^{\abs{X}}})^{(X)} \otimes \qty(\frac{\mathbb{I}}{2^{\abs{Y}}})^{(Y)} \otimes \rho_0^{(Z)}.
\end{gather}
In words, $Y$ is a neighborhood around $X$ that contains the support of all Paulis that intersect with $X$, and $Z$ is the set of all qubits that are not in $X \cup Y$. The state $\rho_0$ is defined such that for all qubits in $Y$, it is the maximally mixed state and for qubits inside $X$, $\rho_0$ is defined in a way such that $\Tr(i\comm{P_i}{P} \rho_0^{(X)}/2)=1$, and for all other qubits, $\rho_0$ can be anything. Then:
\begin{equation}
    \Tr(i \comm{H}{P} \rho_0) = \theta_i .
\end{equation}
\end{lemma}
\begin{proof}
See \cref{sec:master-theorem}.
\end{proof}

This defines a simple algorithm for Hamiltonian learning. For simplicity, for any Pauli $P_i$, we will simply set the observable $P$ to be a single qubit Pauli acting on one site in $X$ such that $\comm{P_i}{P} \neq 0$. 
\begin{algorithm}[H]
\caption{Naive Hamiltonian learning}\label{alg:naive-learn}
\begin{algorithmic}[1]
\Procedure {NaiveInferCoefficients}{$\tau, L, A, N$}
\For{$i \gets 1 \ldots r$}
\State $P \gets \text{single qubit Pauli acting on one site in $X$ where } \comm{P_i}{P} \neq 0$
\State $\rho_0 = \qty(\frac{\mathbb{I} + i \comm{P_i}{P}/2}{2^{\abs{X}}})^{(X)} \otimes \qty(\frac{\mathbb{I}}{2^{\abs{Y}}})^{(Y)} \otimes \rho_0^{(Z)}$ \Comment{$\rho_0^{(Z)}$ is any density matrix}
\State $\tilde{\theta}_i \gets \textsc{EstimateDerivative}(A, L, P, N, \rho_0)$
\EndFor
\EndProcedure
\end{algorithmic}
\end{algorithm}
However, the runtime of this algorithm is $\Omega(r)$, since this procedure must be called once for each term in the Hamiltonian. We propose an improvement of this algorithm wherein we estimate $\Tr(P e^{-iH t} \rho_0 e^{iHt})$ for many different choices of $P$ simultaneously. We aim to set $\rho_0$ in such a way that we can extract coefficients for many terms simultaneously. Yet, rather than using shadow tomography (as done in \cite{franca2022}), which can result in $\order{16^k}$ scaling, we carefully take advantage of our knowledge about the Hamiltonian structure to get a smaller parallelization overhead. The way forward relies on the fact that in \cref{thm:param-select}, $\rho_0^{(Z)}$ can be anything. Similarly to \citep{tang2021}, we partition the terms of our Hamiltonian into groups of terms that can each be inferred simultaneously. This partition is based on a graph coloring, which we define below.

\begin{definition}[Squared graph]\label{defn:sq-graph}
Let the square of the interaction graph, $\mathcal{G}^2$, be the graph with the same vertex set as $\mathcal{G}$ and in which any two vertices are connected if their distance in $\mathcal{G}$ is at most 2. In words, the edges for $\mathcal{G}^2$ are
\begin{equation}
    \qty{(i,k) \mid \exists j \ \qty(\supp P_i \cap \supp P_j \neq \varnothing) \land \qty(\supp P_j \cap \supp P_k \neq \varnothing) \land (i \neq k)}
\end{equation}
\end{definition}

Our algorithm will rely on a graph coloring of $\mathcal{G}^2$. The essential idea is that for Paulis of the same color, there is always a ``moat" separating them. This moat will then be filled with maximally mixed states, which completely suppresses the influence of terms that we are not interested in. A partitioning of the Hamiltonian terms via some $C$-coloring of $\mathcal{G}^2$ makes it natural to rewrite the Hamiltonian using a double sum notation:
\begin{equation}
    H = \sum_{i=1}^C \sum_{j=1}^{\abs{\mathbf{V}_i}} \theta_{i,j} P_{i,j},
\end{equation}
where $\mathbf{V}_i$ is the set of all Paulis with the same color $C_i$. For instance, see \cref{ex:sq-coloring} for a coloring of the squared interaction graph for a 9-qubit TFIM.

\begin{lemma}[Simultaneous inference for a partition]\label{thm:simul-inf}
Let $\mathbf{V}_i$ be a partition in a coloring of $\mathcal{G}^2$. The coefficient for each Pauli in $\mathbf{V}_i$ can be inferred with up to an error $\epsilon \norm{\Theta}_\infty$, with failure probability for each \textbf{individual} coefficient being at most $\delta$ (so the overall failure probability is upper bounded by $\delta \abs{\mathbf{V}_i}$). This can be done with query complexity
\begin{equation}
    \order{\mathscr{D}^2 \log(1/\delta) \polylog(\mathscr{D}/\epsilon) \epsilon^{-2}}. \label{eqn:partition-complexity}
\end{equation}
\end{lemma}
\begin{proof}
See \cref{sec:master-theorem}.
\end{proof}

\begin{theorem}[Hamiltonian learning with unitary dynamics]\label{thm:master-theorem}
Fix a sparsely interacting Hamiltonian $H$ that has $r$ terms in its Pauli expansion with coefficients $\Theta$. For the appropriate choice of Chebyshev degree $L$ and evolution time $A$, \cref{alg:partition-learn} solves the quantum Hamiltonian learning problem (with an additive error $\epsilon \norm{\Theta}_\infty$ and failure probability at most $\delta$) with query complexity
\begin{equation}
    \order{\frac{\mathscr{D}^4 \log (r/\delta) \polylog(\mathscr{D}/\epsilon)}{\epsilon^{2}}}, \label{eqn:total-query}
\end{equation}
and classical processing time complexity
\begin{equation}
    \order{\frac{\mathscr{D}^2 r \log (r/\delta) \polylog(\mathscr{D}/\epsilon)}{\epsilon^{2}}}. \label{eqn:total-time}
\end{equation}
\end{theorem}
\begin{proof}
We find a coloring of the squared interaction graph. There is an efficient way to do this with at most $\mathscr{D}^2$ colors (see \cref{defn:color}). Now, we apply \cref{thm:simul-inf} to each of these partitions. For the detailed proof, see \cref{sec:master-theorem}.
\end{proof}

\begin{algorithm}[H]
\caption{Hamiltonian learning with unitary dynamics}\label{alg:partition-learn}
\begin{algorithmic}[1]
\Procedure {PartitionInferCoefficients}{$\tau, \mathcal{G}, N, L, A, K$}
\State{$\qty{\mathbf{V}_i} \gets  \textsc{GraphColor}(\mathcal{G}^2)$} \Comment{Find $\mathscr{D}^2+1$ partitions of $\mathcal{G}^2$}
\For{$i \gets 1, \ldots, \mathscr{D}^2+1$}
\For{$j \gets 1, \ldots, \abs{\mathbf{V}_i}$} \Comment{Define the observables and states (\cref{thm:simul-inf})}
\State{$P_j' \gets $ a single-qubit Pauli such that $\comm{P_j'}{P_j} \neq 0$}
\State{$\rho_0^{(\supp P_j)} = (\mathbb{I} + i \comm{P_j}{P_j'}/2)/2$}
\EndFor
\For{$q \in \qty(\supp \mathbf{V}_i)'$} \Comment{Fill the moats}
\State{$\rho_0^{(q)} \gets \mathbb{I}/2$}
\EndFor
\For{$k \gets 1, \ldots K$}
\For{$\ell \gets 1, \ldots L$} \Comment{Construct the dataset (\cref{defn:dataset})}
\State $z_\ell \gets -\cos(\frac{2i-1}{L} \pi)$
\State $t_\ell \gets \frac{A}{2}(1+z_\ell)$ 
\State {$\mathsf{M}_\ell \gets N$ simultaneous measurements of $\Tr(P_j' e^{-iH t_{\ell}} \rho_0 e^{iH t_{\ell}})$ for $j \in \qty{1,\ldots,\abs{\mathbf{V}_i}}$}
\EndFor
\For{$j \gets 1, \ldots, \abs{\mathbf{V}_i}$} \Comment{Estimate the first commutator (\cref{alg:first-comm}) for each Pauli in $\mathbf{V}_i$}
\State {$y_\ell \gets $ estimate of $\Tr(P_j' e^{-iH t_{\ell}} \rho_0 e^{iH t_{\ell}})$ by averaging over $\mathsf{M}_\ell$ for $\ell=1, \ldots, L$}
\For{$m \gets 1, \ldots, L-1$} \Comment{Estimate the Chebyshev coefficients (\cref{thm:cheb-fit})}
\State $\tilde{b}_m \gets \frac{2}{L}\sum_{\ell=1}^{L} y_\ell T_m(z_\ell)$ 
\EndFor
\State{$\tilde{\theta}_{i,j}^{(k)} \gets -\frac{2}{A} \sum_{m=1}^{L-1} (-1)^m \tilde{b}_m m^2$} \Comment{$k$th estimate for the coefficient $\theta_{i,j}$}
\EndFor
\EndFor
\For{$j \gets 1, \ldots, \abs{\mathbf{V}_i}$} \Comment{Calculate median of means for each coefficient}
\State $\tilde{\theta}_{i,j} \gets \text{median}_{k=1,\ldots,K}\qty{\tilde{\theta}_{i,j}^{(k)}}$
\EndFor
\EndFor
\EndProcedure
\end{algorithmic}
\end{algorithm}

In a different setup, we may be given access to copies of a Gibbs state at a temperature $\beta^{-1}$. If we measure an observable $P_i$, the expectation will be
\begin{equation}
    \expval{P_i}_\beta = \frac{\Tr(P_i \exp(-\beta H))}{\Tr(\exp(-\beta H))} \label{eqn:gibbs-expec}
\end{equation}
In what follows, we apply the analysis of \citet{tang2021} to formulate $\expval{P_i}_\beta$ as a polynomial in $\beta$, in accordance to the framework in \cref{eqn:polynomial-pic}. We will show that we can learn the coefficients of the Hamiltonian from the first order term in this polynomial, therefore mapping the problem of Hamiltonian learning from Gibbs states onto Hamiltonian learning with unitary dynamics.

\begin{theorem}[Hamiltonian learning with Gibbs states]
The Hamiltonian learning problem (with an additive error $\epsilon \norm{\Theta}_\infty$ and failure probability at most $\delta$) can be solved using
\begin{equation}
    \order{\frac{\mathscr{D}^5 \log (r/\delta) \polylog(\mathscr{D}/\epsilon)}{\epsilon^{2}}} \label{eqn:gibbs-query2}
\end{equation}
copies of the Gibbs state. This can be achieved with a time complexity
\begin{equation}
    \order{\frac{\mathscr{D}^4 r \log(1/\delta) \polylog(\mathscr{D}/\epsilon)}{ \epsilon^{2}}}.
\end{equation}
\end{theorem}
\begin{proof}
The protocol is a near mirror image of the Hamiltonian learning protocol using unitary dynamics. For the full proof, see \cref{sec:ham-gibbs}.
\end{proof}

\section{Numerical Results and Discussion}\label{chap:results}

In this section, we demonstrate numerically the performance of \cref{alg:partition-learn} on a simulated Hamiltonian learning problem using unitary dynamics. We will apply \cref{alg:partition-learn} for learning a transverse field Ising model:
\begin{equation}
    H = \sum_{i=1}^{n-1} J_i \sigma_z^{(i)} \otimes \sigma_z^{(i+1)} + \sum_{i=1}^n B_i \sigma_x^{(i)},
\end{equation}
where $J_i, B_i \sim \text{Unif}(-1,1)$. To simulate the dynamics of this Hamiltonian, we use the time-evolution block-decimation method \citep{zwolak2004,paeckel2019,white2004,daley2004,vidal2004}.


There are two hyperparameters in \cref{alg:partition-learn} that are crucial for determining the performance of the Hamiltonian learning routine: the maximum evolution time $A$ and the Chebyshev degree $L$. Setting these parameters is a delicate balance between noise-induced error and modelling errors. If $A$ is too low or $L$ is too high, the variance in the dataset will dominate the error, and on the other hand, if $A$ is too high or $L$ is too low, the modelling error will dominate. It is generally desirable to set these two parameters such that the modelling and noise errors are comparable.\footnote{However, in some settings, it may be desirable to let the dataset variance grow somewhat larger than the modelling error, since this error can be quantified exactly via \cref{eqn:var-est}, where $\sigma_\ell^2$ can be obtained by a bootstrap estimate from the dataset. There are no similar methods to quantify the modelling error.} A~possible method for setting $A$ and $L$ can be to optimize the error bounds (see \cref{fig:hyperparam-settings}). Numerically, these optimal values behave as anticipated in \cref{thm:err-scale}: the optimal $L^*$ scales with $\order{\log \epsilon^{-1}}$, and $A/\tau \sim 1$ which leads to shot requirements scaling with $\order{\text{polylog}(1/\epsilon)\epsilon^{-2}}$. 
\begin{figure}
    \centering
    \includegraphics[width=0.8\textwidth]{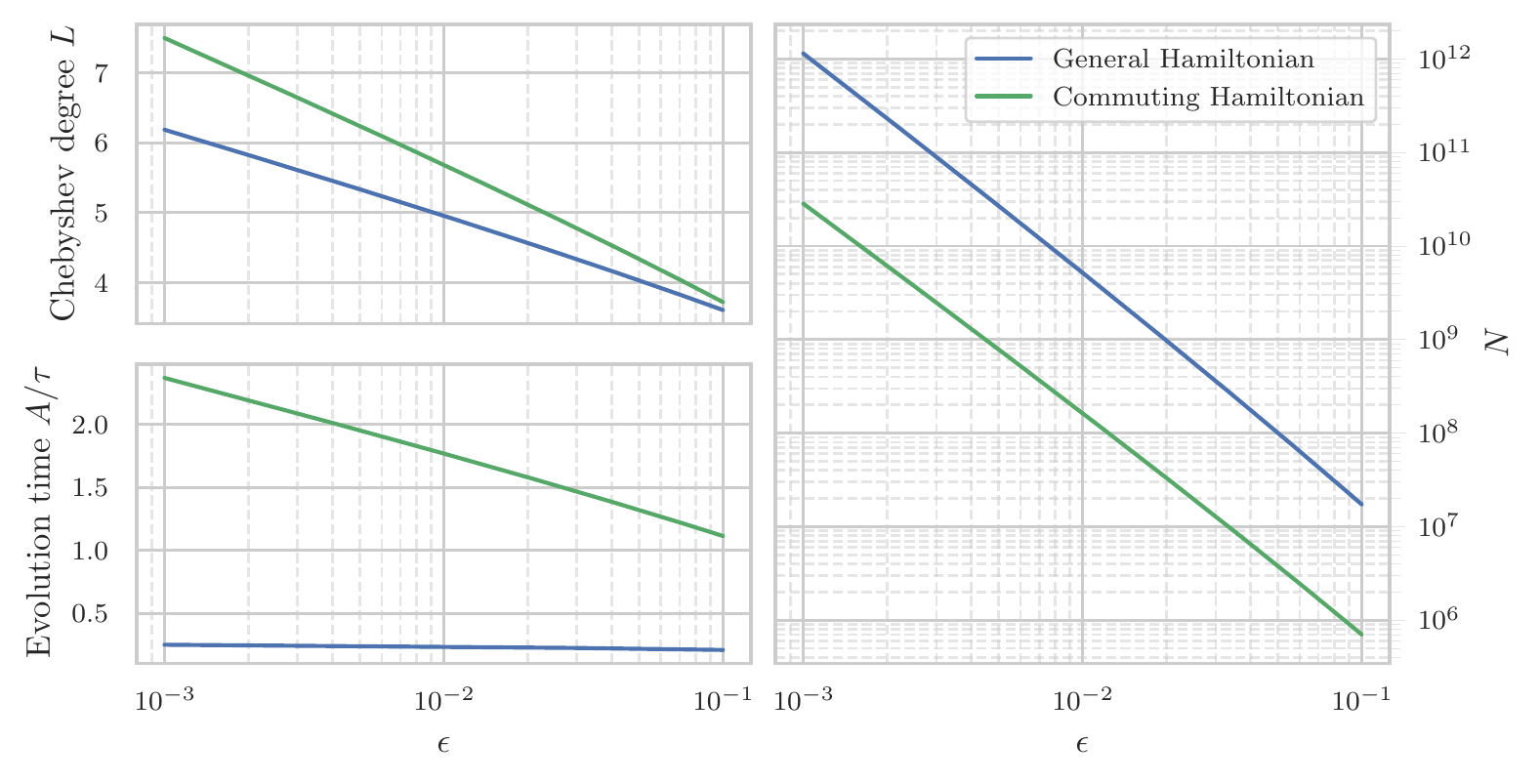}
    \caption[Optimal hyperparameter settings]{Settings for $A$ and $L$ as a function of the desired error $\epsilon$. These settings are found based on minimizing the upper bound on $N \cdot L$ in \cref{eqn:n-bound}, plugging in the two types of derivative bounds found in \cref{lem:norm-bound} for arbitrary and commuting Hamiltonians. For the case of the arbitrary Hamiltonians, note that although we were free to set $A$ and $L$ in the proof of \cref{thm:err-scale}, we indeed find the optimal $L = \order{\log \epsilon^{-1}}$, $A = \order{\tau}$, and $N = \order{\text{polylog}(1/\epsilon)\epsilon^{-2}}$. We find similar scaling for the case of the commuting Hamiltonian in every variable except $A/\tau$, which also to scale as $ \order{\log \epsilon^{-1}}$: this is because the modeling error term is suppressed by a factorial, and scales like $\nicefrac{(A/\tau)^L}{L!}$. Despite this, the overall query complexity is only better than the general case by a constant factor.}
    \label{fig:hyperparam-settings}
\end{figure}

In \cref{fig:model-err}, we show the error in the recovered Hamiltonian parameters corresponding to a target error of $\epsilon=0.021$. As expected, the theoretical prediction for the noise error is close to perfect. However, the modelling error is drastically overestimated by nearly four orders of magnitude. This is reasonable, and originates in the upper bound for the norm of the iterated commutator $\norm{\qty[H^m P]}$ found in \cref{lem:norm-bound}, since the counting bound in \cref{thm:counting} is unlikely to be saturated. This miscalculated modelling error has important consequences for the algorithm, since it results in a poorly specified evolution time $\tau$. Our bounds in \cref{lem:norm-bound} are too loose, hence our calculated optimal evolution time is too small, resulting in a higher noise error than necessary. We propose a number of ways to remedy this, beginning with a heuristic attempt to tighten the bound in \cref{lem:norm-bound}. The reason for the first two optimizations is nonobvious, and is described in detail in \cref{sec:heur}.

\begin{figure}[H]
    \centering
    \includegraphics[width=0.65\textwidth]{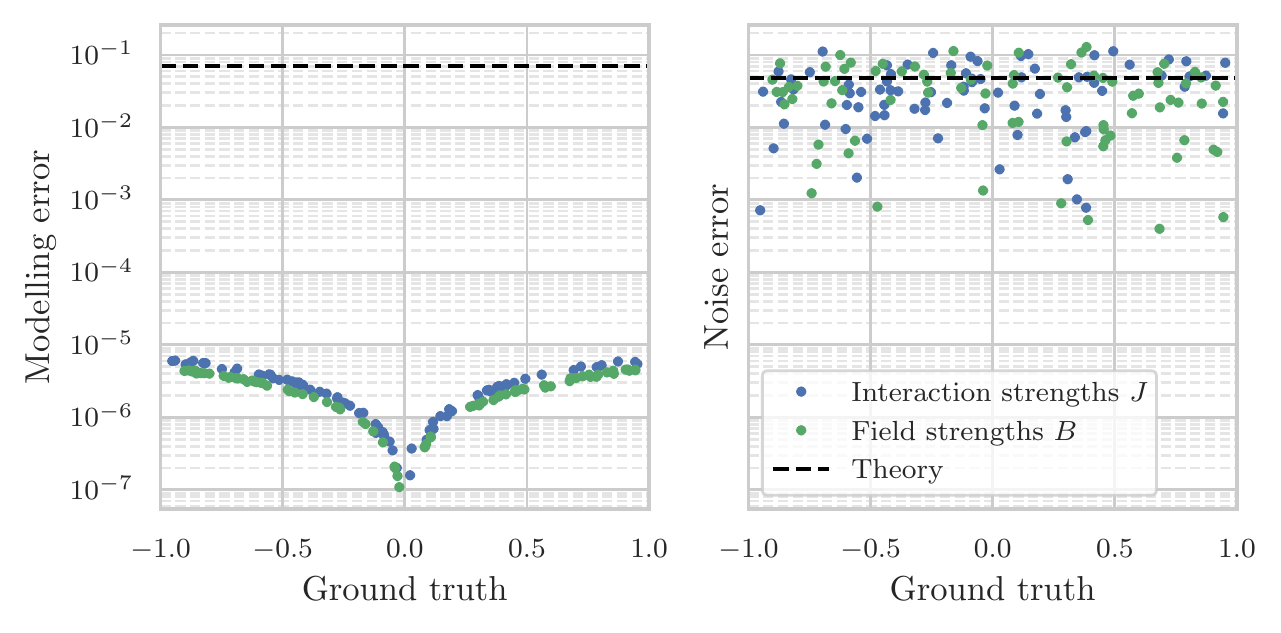}
    \caption[Empirical error of the Hamiltonian learning protocol]{The empirical modelling and noise error of the Hamiltonian learning protocol using the optimal $A$ and $L$ for $\epsilon=0.021$ as prescribed in \cref{fig:hyperparam-settings}. The modelling errors are calculated with a noise-free dataset, and the noise errors are calculated from a single noisy dataset. The dashed line indicates the maximum theoretical modelling error on the left and indicates the predicted variance due to noise on the right.}
    \label{fig:model-err}
\end{figure}

\begin{optimization}\label{opt:1}
When calculating the optimal evolution time $A$ and polynomial degree $L$, we replace $\mathscr{D}$ with
\begin{equation}
    \overline{\mathscr{D}} = \frac{1}{2\abs{V}}\sum_{v \in V} \deg(v).
\end{equation}
\end{optimization}

\begin{optimization}\label{opt:2}
A further improvement of algorithm performance can be gained by optimizing the number of queries allocated to each evaluation of $\expval{P(t)}$. Let 
\begin{equation}
    c_\ell \equiv \qty(\sum_{m=1}^{L-1} (-1)^m m^2 T_m(z_\ell))^2.
\end{equation}
Then, for a fixed number of total queries $N_\textrm{max}$, we allocate a number of queries
\begin{equation}
    N_{\ell} = \frac{\sqrt{c_\ell}}{\sum_\ell \sqrt{c_\ell}} N_\textrm{max}
\end{equation}
when estimating $\expval{P(t_\ell)}$.
\end{optimization}
\begin{optimization}\label{opt:3}
Since we know the polynomial at $t=0$ to have a value of $0$ (our SPAM parameters are set such that this is always true), we constrain our solution for the Chebyshev coefficients such that the constant term in the polynomial is $0$. That is, rather than fitting a generic polynomial $\sum_{\ell=0} c_\ell t^\ell$ to the data, we fit $\sum_{\bm{\ell=1}} c_\ell t^\ell$. This further reduces the variance of our solution, mitigating the noise-induced error.
\end{optimization}

We apply \cref{opt:1,opt:2,opt:3} (the heuristic calculation of $\overline{\mathscr{D}}$, improved shot allocation, and solution constraints), to the TFIM model. The error distributions shown in \cref{fig:err} are the result of running the Hamiltonian learning algorithm on 40 random instances of the 80-qubit TFIM problem. Since we can already calculate the improvements realized by \cref{opt:1} (decrease in error by a factor $\gtrsim \nicefrac{\mathscr{D}}{\overline{\mathscr{D}}}$), we show only the effects of \cref{opt:2,opt:3}. Empirically, we are able to realize significant improvements by applying these optimizations. There is a decrease in error by a factor $\sim 1.7$, enabling us to use $\gtrsim 3$ times fewer shots.

\begin{figure}[H]
    \centering
    \includegraphics[width=0.8\textwidth]{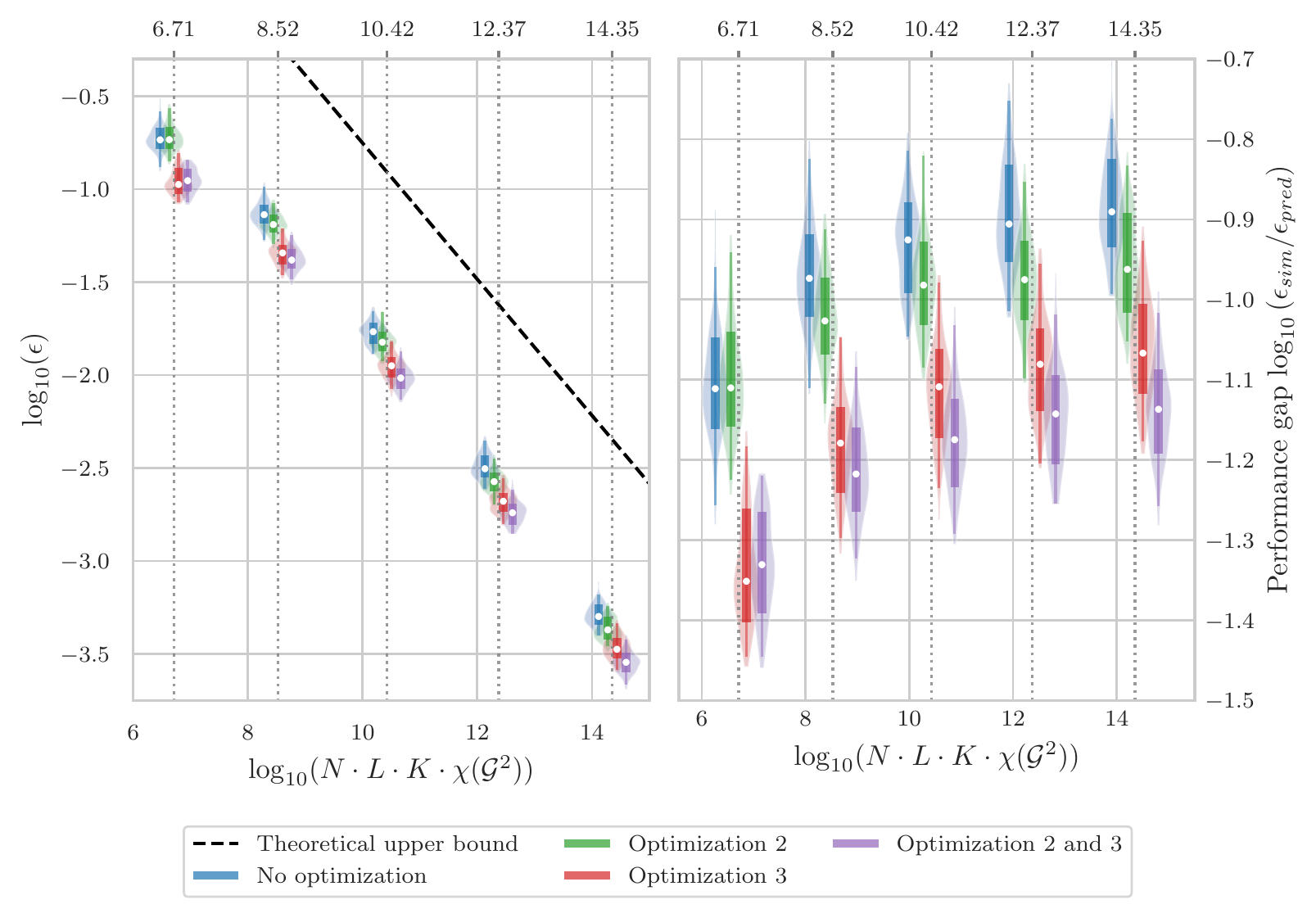}
    \caption[Empirical error distribution]{On the left, we show the maximum absolute error across all 159 coefficients of the 80-qubit TFIM model, plotted against the total number of queries $N \cdot L \cdot K \cdot \chi(\mathcal{G}^2)$, and on the right, we show the difference between the theoretical error upper bound and the empirical errors from numerical simulations (note the log-log scale for both plots). The violin plots show the distribution of maximum absolute errors from 100 random initializations of the TFIM (with coefficients sampled uniformly between $-1$ and $1$). The distributions show the $\qty[1\%,99\%]$ interval in a narrow line, a $\qty[16\%,84\%]$ interval in a wider line, and the median marked in white. The violin plots are offset by a small amount for visualization purposes, but each cluster of four violin plots used the same number of queries marked by the dotted grey lines. 
    Recall $N$ and $K$ are defined in \cref{thm:err-scale}, $L$ is the Chebyshev degree, and $\chi(\mathcal{G}^2)$ is the chromatic number of the squared interaction graph (\cref{defn:color}).
    We set the failure probability to $\delta=15\%$.}
    \label{fig:err}
\end{figure}

As shown on the right of \cref{fig:err}, there is a significant gap (by a factor $\sim 10$) between the theoretical error upper bound (calculated using the effective degree $\overline{\mathscr{D}}$) and the empirical error. We propose two possible reasons for this gap. First, despite the heuristic improvements made in \cref{opt:1}, there is still a significant gap between the theoretical and empirical modelling error (see \cref{fig:model-err-heur}). This adds a discrepancy by a factor $\sim 2$ between the theoretical and empirical errors, since the optimal configurations always have a theoretical error that is composed roughly of half modeling error and half noise error. Secondly, the settings for $N$ and $K$ used in \cref{thm:err-scale} to guarantee $\mathbb{P}(\text{maximum absolute error} > \epsilon) < \delta$ are overestimates of the true requirements.

\section{Conclusions}
In this work, we have discussed the quantum Hamiltonian learning problem. We introduced a unifying model for Hamiltonian learning using both unitary dynamics and Gibbs states. By subsuming these two approaches into the same model, we were able to describe an abstract routine for learning the Hamiltonian of a quantum many-body system given limited access to the system. This routine was based on fixing certain SPAM parameters, then viewing the system as a function $f$ of a single variable (in this work, we consider this variable to be either time $t$ or inverse temperature $\beta$). We argued that the coefficients in the Taylor expansion of $f$ (particularly, the first order coefficient) could be connected with the Hamiltonian parameters in a straightforward manner, then showed that the relevant coefficient could be inferred both accurately and efficiently from noisy evaluations of $f$. Finally, we concluded by describing how our protocol could achieve better than linear query complexity in $r$ (the number of Hamiltonian parameters) by using SPAM configurations amenable to simultaneous measurements. 

This culminated in our main result, wherein we proposed an algorithm that achieves an almost noise-limited ($\sim \frac{\text{polylog}(\epsilon^{-1})}{\epsilon^2}$) query complexity, similar to that of \citet{tang2021} and \citet{franca2022}. However, our work represents an advance for several reasons. In comparison to \citet{tang2021}, we significantly reduce their dependence on the parameter $\mathscr{D}$ from $\mathscr{D}^{21}$ to $\mathscr{D}^4$. In comparison to \citet{franca2022}, we generalize their protocol to learning with Gibbs states and improve their measurement parallelization overhead from $\order{16^k}$ to $\order{\mathscr{D}^2}$ -- this can be a significant improvement when we have prior knowledge of the Pauli terms in the Hamiltonian (i.e., $\mathscr{D} < 4^k$). Furthermore, by deriving explicit bounds on the performance of our algorithm, we were able to provide precise numerical prescriptions for theoretically optimal hyperparameters such as maximum evolution time and Chebyshev degree. We concluded by proposing a number of heuristic improvements to our algorithm, and argued they were reasonable to apply in general. This combination of improvements makes our algorithm applicable in practice for non-trivial Hamiltonians, which we demonstrated by applying our protocol on a simulated, large (80-qubit) problem.

Although we have demonstrated a successful application of our learning algorithm on a simulated problem, this simulation did not include possible detrimental experimental effects. Our algorithm makes minimal SPAM requirements (requiring only single-qubit measurements and simple product states), but further investigation is needed to determine the effects of SPAM errors on performance. We also leave for later works a study of how this protocol can be improved by making stronger assumptions on either the Hamiltonian or the suite of interactions available to us. For instance, we already showed a constant (but significant) drop in the number of measurements required for learning a commuting Hamiltonian with unitary dynamics. We expect a similar effect for Hamiltonian learning with Gibbs states. Furthermore, if we assume we can interact with our system using a trusted quantum simulator of our own, a variety of approaches become possible. Among these is Hamiltonian learning with Loschmidt echoes, as done in \citet{wiebe2014a}. Rigorous performance bounds have not yet been found for this approach, but we speculate that a similar application of our techniques may yield improved performance -- however, we leave this for future works.

\begin{acknowledgments}
The authors thank Hsin-Yuan Huang (Robert), Matthias C.~Caro, Diego Garc\'ia-Mart\'in, and Marco Cerezo for inspiring discussions and for comments on an earlier draft. AG acknowledges support from the U.S. Department of Energy (DOE) through a quantum computing program sponsored by the Los Alamos National Laboratory (LANL) Information Science \& Technology Institute. LC was supported by the Laboratory Directed Research and Development (LDRD) program of LANL under project number 20210116DR. PJC was supported by the LANL ASC Beyond Moore's Law project. 
\end{acknowledgments}

\clearpage
\appendix
\section{Chebyshev Regression}\label{sec:cheb}
To fit a black-box function $f$ where we can query $f(x)$ for $x$ within some window $[A,B]$, we are often interested in approximating $f$ with a polynomial of degree $L$. This is known as polynomial interpolation. When we have complete freedom in choosing the location of the points $x \in [A,B]$, a popular method is Chebyshev interpolation. There are a number of favorable properties associated with this method. Importantly, it achieves close to an optimal approximation error on $[A,B]$. That is, if $\tilde{f}$ is the degree-$L$ polynomial resulting from Chebyshev interpolation, $\max_{x \in [A,B]} \abs{f(x)-\tilde{f}(x)}$ is close to the minimal possible value among all polynomials of degree $L$ \citep{mason2003,press2007}.

Chebyshev interpolation takes its name from its extensive use of a class of polynomials known as Chebyshev polynomials. This class of polynomials is unique in that for $z \in [-1,1]$, they can be written:
\begin{equation}
    T_n(z) = \cos(n \arccos z)
\end{equation}
Although this is not in polynomial form, by using the identity
\begin{equation}
    \cos(n \theta) = \sum_{\substack{r=0\\2r \leq n}} (-1)^r \binom{n}{2r} \cos^{n-2r}(\theta) \sin^{2r}(\theta),
\end{equation}
we find:
\begin{align}
    T_n(z) &= \sum_{\substack{r=0\\2r \leq n}} (-1)^r \binom{n}{2r} \cos^{n-2r}(\arccos z) \sin^{2r}(\arccos z) \\
    &= \sum_{\substack{r=0\\2r \leq n}} (-1)^r \binom{n}{2r} z^{n-2r} (1-z^2)^r \label{eqn:tn-expansion}.
\end{align}
This makes clear that the degree of $T_n$ is $n$. The first seven Chebyshev polynomials are
\begin{align}
\begin{split}
    T_0(z) &= 1 \\
    T_1(z) &= z \\
    T_2(z) &= 2z^2-1 \\
    T_3(z) &= 4z^3 - 3z \\
    T_4(z) &= 8z^4 - 8z^2 + 1 \\
    T_5(z) &= 16z^5 - 20 z^3 + 5z \\
    T_6(z) &= 32z^6 - 48 z^4 + 18 z^2 - 1.
\end{split}
\end{align}

Below, we list a number of useful properties of Chebyshev polynomials.
\begin{lemma}\label{lem:cheb-roots}
The roots of the $n$th Chebyshev polynomial are
\begin{equation}
    z_m = -\cos(\frac{2m-1}{2n} \pi) \qc m=1,\ldots,n
\end{equation}
\end{lemma}
\begin{proof}
This follows from a simple substitution. 
\begin{align*}
    T_n(z_m) &= \cos(n \arccos(-\cos(\frac{2m-1}{2n}) \pi)) \\
    &= \cos(-\qty(m-\frac{1}{2}) \pi) \\
    &= 0
\end{align*}
\end{proof}

\begin{lemma}\label{lem:disc-orth}
If $\qty{z_m \mid m=1, \ldots, n}$ are the roots of $T_n$, then for any $i, j \leq n$:
\begin{equation}
    \sum_{m=1}^n T_i(z_m) T_j(z_m) = \begin{cases}
    0 & \text{if $i \neq j$} \\
    \frac{n}{2} & \text{if $i=j \neq 0$} \\
    n & \text{if $i=j=0$}
    \end{cases}
\end{equation}
In the general case, if either $i > n$ or $j > n$: 
\begin{equation}\label{eqn:gen-dist-orth}
    \sum_{m=1}^n T_i(z_m) T_j(z_m) = \begin{cases}
    \frac{n}{2} & \text{if either $i + j \mid 2n$ or $i-j \mid 2n$} \\
    n & \text{if $(i+j) \mid 2n$ and $i-j \mid 2n$} \\
    0 & \text{otherwise}
    \end{cases}
\end{equation}
\end{lemma}
\begin{proof}
The case of $i=j=0$ is trivial, since $T_0=1$. So, below, we will assume either $i \neq 0$ or $j \neq 0$.
\begin{align*}
    \sum_{m=1}^n T_i(z_m) T_j(z_m) &= \sum_{m=1}^n \cos(\frac{i (2m-1)}{2n} \pi) \cos(\frac{j (2m-1)}{2n} \pi)
    \intertext{Using the identity $\cos(a) \cos(b) = \frac{\cos(a-b) + \cos(a+b)}{2}$:}
    &= \frac{1}{2} \qty(\sum_{m=1}^n \cos(\frac{(i-j)(2m-1)}{2n} \pi) + \cos(\frac{(i+j)(2m-1)}{2n} \pi))
    \intertext{We use the identity $\sum_{m=1}^n \cos((m-1/2)\theta) = \frac{\sin(n \theta)}{2 \sin(\theta/2)}$, with the understanding that if $\theta=0$, this evaluates to $n$ (using the identity $\lim_{\theta \rightarrow 0} \frac{\sin(n \theta)}{\sin \theta}=n$).}
    &= \frac{1}{2} \qty(\frac{\sin((i-j)\pi)}{2\sin(\frac{i-j}{2n} \pi)} + \frac{\sin((i+j)\pi)}{2\sin(\frac{i+j}{2n} \pi)})
    \intertext{For $i, j \leq n$, this reduces to:}
    &= \begin{cases}
    0 & \text{if $i \neq j$} \\
    \frac{n}{2} & \text{if $i=j \neq 0$} \\
    n & \text{if $i=j=0$}
    \end{cases}
\end{align*}
The general case follows identically.
\end{proof}

With these preliminaries, we now describe Chebyshev interpolation. After fixing a degree $L-1$ for the interpolating polynomial, we evaluate $f$ at the roots of $T_L(z)$.\footnote{This assumes the window in which we are approximating $f$ is $\qty[-1,1]$, however, the generalization to arbitrary windows is straightforward.} This results in a dataset $\mathcal{D}=\qty{(z_i, y_i) \mid i = 1,\ldots,L}$ where $z_i$ are the roots of $T_L(z)$ and $y_i=f(z_i)$. The functional form with which we interpolate $\qty{y_i \mid i=1,\ldots,L}$ is simply a linear combination of the first $L$ Chebyshev polynomials:
\begin{equation}
    \tilde{f}(z; \qty{b_i}) = \sum_{\ell=0}^{L-1} b_\ell T_\ell(z)
\end{equation}

\begin{theorem}[Chebyshev interpolation]\label{thm:cheb-fit}
The coefficients $\qty{b_i}$ such that $\tilde{f}(z; \qty{b_i})$ perfectly fits our dataset are:
\begin{equation}\label{eqn:cheb-coeff}
    b_\ell = \begin{cases}
        \frac{1}{L} \sum_{i=1}^L y_i & \text{for $\ell=0$} \\
        \frac{2}{L} \sum_{i=1}^L y_i T_\ell(z_i) & \text{for $\ell \neq 0$}
    \end{cases}
\end{equation}
\end{theorem}
\begin{proof}
We minimize the squared distance squared error between $\tilde{f}$ and the dataset. Then, since any set of $L$ points can be perfectly fitted by a degree $L-1$ polynomial, the resulting coefficients will be the interpolation solution. We set $\sum_i \pdv{b_\ell} \qty(y_i - \tilde{f}(z_i; B))^2 = 0$:
\begin{align*}
    0 &= \sum_{i=1}^L \qty(y_i - \tilde{f}(z_i; B)) T_\ell(z_i) \\
    \sum_{i=1}^L y_i T_\ell(z_i) &= \sum_{m=0}^{L-1} b_m \sum_{i=1}^L T_\ell(z_i) T_m(z_i)
    \intertext{We now use \cref{lem:disc-orth}. If $\ell=0$:}
    \sum_{i=1}^L y_i T_\ell(z_i) &= L \sum_{m=0}^{L-1} b_m \delta_{\ell m} \\
    b_0 &= \frac{1}{L} \sum_{i=1}^L y_i
    \intertext{This simply says that the constant term in the fitted polynomial is the average over $\qty{y_i}$, an intuitive result. For $\ell \neq 0$:}
    b_\ell &= \frac{2}{L} \sum_{i=1}^L y_i T_\ell(z_i)
\end{align*}
\end{proof}

\section{Proof of \texorpdfstring{\cref{lem:norm-bound}}{Theorem \ref{lem:norm-bound}}}\label{sec:bound-iter}

\begin{definition}[Types of tuples]
\label{defn:tuples}
We introduce some nomenclature that abbreviates much of our later proofs. First, let $\Theta$ be a tuple of the Hamiltonian parameters $(\theta_1, \ldots, \theta_r)$. Also, we introduce multi-index sets $\alpha \in \qty(\mathbb{Z}_{\geq 0})^r$, which are tuples of nonnegative integers $(\alpha_1, \alpha_2, \ldots, \alpha_r)$. We define:
\begin{equation}
\abs{\alpha} = \sum_i \alpha_i \qc \Theta^{\alpha} = \prod_i \theta_i^{\alpha_i}
\end{equation}
We call $\abs{\alpha}$ the size of $\alpha$.
Finally, we define ``term tuples" $S \in \qty{1, \ldots, r}^m$ of size $m$ to be ordered tuples $(s_1, \ldots, s_m)$. Each entry in $S$ is meant to point at a term in the Hamiltonian. 

There is a correspondence between size $m$ multi-index sets and term tuples. A term tuple $S$ with size $m$ maps to $\alpha$ as follows: $\qty(\alpha(S))_i = \abs{\qty{j \mid S_j = i, j=1,\ldots,m}}$. This mapping essentially counts the number of occurrences of $i$ in $S$.
Graphically, this is represented in \cref{fig:s-def}. This mapping is surjective, in the sense that any multi-index set $\alpha$ with size $m$ can be written as $\alpha(S)$ for some $S$ with size $m$. However, it is not injective, since multiple term tuples can map to the same multi-index set.

\begin{figure}[H]
    \centering
    \begin{tikzpicture}[node distance={25mm}, thick, main/.style = {rectangle, minimum height=1cm}
    ] 
\node[main] (0) {$H$}; 
\node[main] (1) [right=0mm of 0, text height=3mm] {$=$};
\node[main] (2) [right=25mm of 1] {$\theta_{\textcolor{red}{1}} P_{\textcolor{red}{1}}$};
\node[main] (3) [right=25mm of 2] {$+$};
\node[main] (4) [right=0mm of 3] {$\theta_{\textcolor{blue}{2}} P_{\textcolor{blue}{2}}$};
\node[main] (5) [right=0mm of 4] {$+$};
\node[main] (6) [right=10mm of 5] {$\theta_{\textcolor{purple}{3}} P_{\textcolor{purple}{3}}$};
\node[main] (7) [right=10mm of 6] {$+$};
\node[main] (8) [right=0mm of 7] {$\ldots$};
\node[main] (9) [right=0mm of 8] {$+$};
\node[main] (10) [right=0mm of 9] {$\theta_{\textcolor{green}{r}} P_{\textcolor{green}{r}}$};

\node[main] (11) [below=5mm of 0] {$\alpha$};
\node[main] (12) [below=5mm of 1, text height=2mm] {$=$};
\node[main] (13) [below=5mm of 2] {$\alpha_{\textcolor{red}{1}}=3$};
\node[main] (14) [below=5mm of 3, text height=4mm] {$,$};
\node[main] (15) [below=5mm of 4] {$\alpha_{\textcolor{blue}{2}}=0$};
\node[main] (16) [below=5mm of 5, text height=4mm] {$,$};
\node[main] (17) [below=5mm of 6] {$\alpha_{\textcolor{purple}{3}}=2$};
\node[main] (18) [below=5mm of 7, text height=4mm] {$,$};
\node[main] (19) [below=5mm of 8, text height=4mm] {$\ldots$};
\node[main] (20) [below=5mm of 9, text height=4mm] {$,$};
\node[main] (21) [below=5mm of 10] {$\alpha_{\textcolor{green}{r}}=1$};

\node[main] (22) [below=5mm of 11] {$S$};
\node[main] (23) [below=5mm of 12, text height=2mm] {$=$};
\node[main] (24) [below left=5mm and 2mm of 13] {$s_1=\textcolor{red}{1}$};
\node[main] (25) [right=0mm of 24, text height=4mm] {$,$};   
\node[main] (26) [below=5mm of 13] {$s_2=\textcolor{red}{1}$};
\node[main] (27) [right=0mm of 26, text height=4mm] {$,$};   
\node[main] (28) [below right=5mm and 2mm of 13] {$s_3=\textcolor{red}{1}$};
\node[main] (29) [right=0mm of 28, text height=4mm] {$,$};   

\node[main] (30) [below left=5mm and -5mm of 17] {$s_4={\textcolor{purple}{3}}$};
\node[main] (31) [right=0mm of 30, text height=4mm] {$,$};   
\node[main] (32) [below right=5mm and -5mm of 17] {$s_5={\textcolor{purple}{3}}$};
\node[main] (33) [below=5mm of 18, text height=4mm] {$,$};   
\node[main] (34) [below=5mm of 19, text height=4mm] {$\ldots$};   
\node[main] (35) [below=5mm of 20, text height=4mm] {$,$};   
\node[main] (36) [below=5mm of 21] {$s_m={\textcolor{green}{r}}$};

\draw[<-] (13) -- (24);
\draw[<-] (13) -- (26);
\draw[<-] (13) -- (28);

\draw[<-] (17) -- (30);
\draw[<-] (17) -- (32);

\draw[<-] (21) -- (36);
\end{tikzpicture} 
    \caption[The correspondence of $s_1,\ldots,s_m$ with $\alpha_1,\ldots,\alpha_n$]{The correspondence of $s_1,\ldots,s_m$ with $\alpha_1,\ldots,\alpha_n$. As the colors indicate, the value of each $s_i$ indicate a term in the Hamiltonian (e.g., $s_i=\mathbf{3}$ refers to the $\theta_{\mathbf{3}} P_{\mathbf{3}}$ term).}
    \label{fig:s-def}
\end{figure}

\end{definition}

\begin{lemma}
The iterated commutator $\qty[H^m P]$ can be expanded as:
\begin{equation}
    \qty[H^m P] = \sum_{S \in \qty{1, \ldots, r}^m} \Theta^{\alpha(S)} \qty[P_{s_1}, [P_{s_2}, \ldots, [P_{s_m}, P]]] .\label{eqn:poly-expand}
\end{equation}
\textit{We emphasize that \cref{eqn:poly-expand} is nothing more than a polynomial expansion, in the Hamiltonian coefficients, of the iterated commutator.}
\end{lemma}
\begin{proof}
We inductively arrive at the following formula for the iterated commutator $\comm{H^m}{P}$. 
\begin{equation}
    \qty[H^m P] = \sum_{s_1=1}^r \sum_{s_2=1}^r \underset{m \text{ times}}{\ldots} \sum_{s_m=1}^n \theta_{s_1} \theta_{s_2} \ldots \theta_{s_m} \qty[P_{s_1}, [P_{s_2}, \ldots, [P_{s_m}, P]]] \label{eqn:naive-iter}
\end{equation}
This follows simply from the linearity of the commutator (i.e., $\comm{A+B}{C}=\comm{A}{C}+\comm{B}{C}$): 
\begin{align*}
    \qty[H^{m+1} P] &= \comm{H}{\qty[H^m P]}
    \intertext{By the induction hypothesis:}
    &= \comm{H}{\sum_{s_1=1}^r \sum_{s_2=1}^r \ldots \sum_{s_m=1}^r \theta_{s_1} \theta_{s_2} \ldots \theta_{s_m} \qty[P_{s_1}, [P_{s_2}, \ldots, [P_{s_m}, P]]]}
    \intertext{By linearity of the commutator:}
    &= \sum_{s_{m+1}=1}^r \sum_{s_1=1}^r \ldots \sum_{s_m=1}^r \theta_{s_1} \ldots \theta_{s_m} \theta_{s_{m+1}} \qty[P_{s_{m+1}}, [P_{s_1}, \ldots, [P_{s_m}, P]]]
\end{align*}
as desired, up to relabeling the indices. These tuples are members of $(\mathbb{Z}_n)^m$, so we rewrite \cref{eqn:poly-expand} as:
\begin{align*}
    \qty[H^m P] &= \sum_{S \in \qty{1, \ldots, r}^m} \theta_{s_1} \ldots \theta_{s_m} \qty[P_{s_1}, [P_{s_2}, \ldots, [P_{s_m}, P]]] \\
    &= \sum_{S \in \qty{1, \ldots, r}^m} \Theta^{\alpha(S)} \qty[P_{s_1}, [P_{s_2}, \ldots, [P_{s_m}, P]]]
\end{align*}
\end{proof}

Without any structural assumptions about the Hamiltonian $H$, the form of this expansion is not useful. However, by applying our assumption that the Hamiltonian is sparsely interacting, we can find useful bounds using the above expansion.

\begin{definition}[Support tree]\label{defn:supp-tree}
Let $P$ be any Pauli operator such that there is some $P_i$ where $\supp P \subseteq \supp P_i$ (for the remainder of this work, wherever we write $P$, it will always satisfy this assumption). For this Pauli, we can define the support tree induced by $P$, which we call $\mathcal{T}_P$ (see \cref{fig:supp-tree}). This tree is defined as follows: the children of any node $P_i$ is $$\text{children}(P_i) = \qty{P_j \mid \supp(P_i) \cap \supp(P_j) \neq \varnothing}.$$ The root of $\mathcal{T}_P^{(m)}$ is $P$.
\end{definition}
\begin{example}
The following is a support tree induced by $P=\sigma_x^{(2)}$ on the 9-qubit TFIM shown in \cref{ex:tfim-model}.
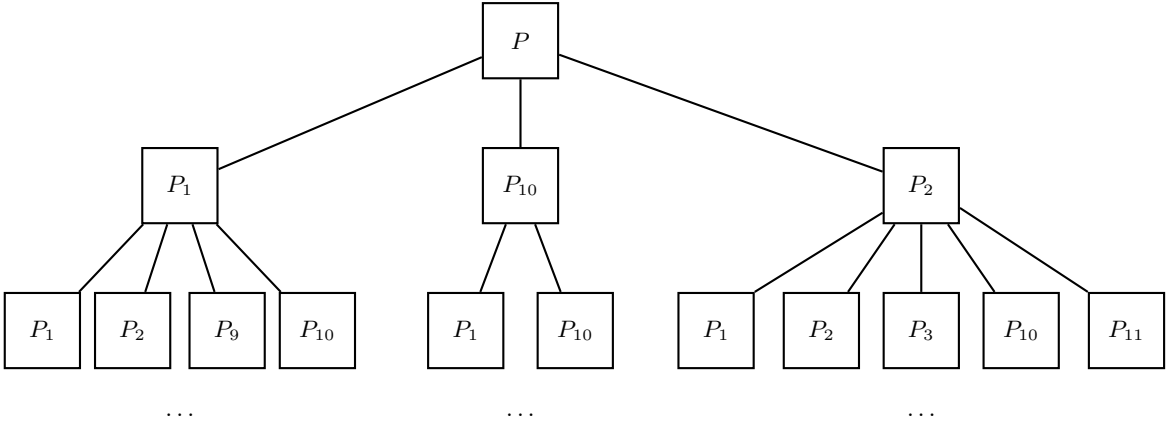
\begin{figure}[H]
    \centering
    \begin{tikzpicture}[node distance={30mm}, thick, main/.style = {draw, rectangle, minimum size=1cm}
    ] 
\node[main] (0) {$P$}; 
\node[main] (1) [below left=9mm and 35mm of 0] {$P_1$};
\node[main] (2) [below=9mm of 0] {$P_{10}$}; 
\node[main] (3) [below right=9mm and 43mm of 0] {$P_{2}$}; 

\node[main] (5) [below left=9mm and 8mm of 1] {$P_1$}; 
\node[main] (6) [below left=9mm and -4mm of 1] {$P_2$}; 
\node[main] (7) [below right=9mm and -4mm of 1] {$P_{9}$}; 
\node[main] (8) [below right=9mm and 8mm of 1] {$P_{10}$}; 

\node[main] (9) [below left=9mm and -3mm of 2] {$P_1$}; 
\node[main] (10) [below right=9mm and -3mm of 2] {$P_{10}$};

\node[main] (12) [below left=9mm and 17mm of 3] {$P_1$}; 
\node[main] (13) [below left=9mm and 3mm of 3] {$P_{2}$}; 
\node[main] (14) [below=9mm of 3] {$P_3$}; 
\node[main] (15) [below right=9mm and 3mm of 3] {$P_{10}$}; 
\node[main] (16) [below right=9mm and 17mm of 3] {$P_{11}$}; 

\draw[-] (0) -- (1);
\draw[-] (0) -- (2);
\draw[-] (0) -- (3);
\draw[-] (1) -- (5);
\draw[-] (1) -- (6);
\draw[-] (1) -- (7);
\draw[-] (1) -- (8);
\draw[-] (2) -- (9);
\draw[-] (2) -- (10);
\draw[-] (3) to [in=45, out=215,looseness=0] (12);
\draw[-] (3) -- (13);
\draw[-] (3) -- (14);
\draw[-] (3) -- (15);
\draw[-] (3) to [in=135, out=330,looseness=0] (16);
\node (17) [below=24mm of 1] {$\ldots$};
\node (18) [below=24mm of 2] {$\ldots$};
\node (19) [below=24mm of 3] {$\ldots$};
\end{tikzpicture} 
    \caption[The support tree $\mathcal{T}_P$]{The support tree $\mathcal{T}_P$ for the TFIM system shown in \cref{fig:int-graph}, induced by some Pauli $P$ that acts on qubit 2 (e.g., $P=\sigma_x^{(2)}$). For illustration purposes, we truncate the tree at a depth of 2. Each node has at most $\mathscr{D}+1$ children. The structure of this tree is entirely dependent on the interaction graph (see \cref{defn:inter-graph}).}
    \label{fig:supp-tree}
\end{figure}

\end{example}

\begin{theorem}[Counting labeled subtrees]\label{thm:counting}
Let $\mathcal{T}$ be an infinite rooted $d$-regular tree (i.e., every node has $d$ children). There are exactly
\begin{equation}
    a_m = \prod_{n=1}^{m-1} (n(d-1)+1) \leq (d-1)^{m-1} m!
\end{equation}
labelled rooted subtrees in $\mathcal{T}$, of size $m$ and labels $\qty{1,\ldots,m}$, such that the label of any given node is less than that of all of its descendants.
\end{theorem}
\begin{proof}
Let $a_m$ be the number of labelled rooted subtrees of size $m$ that satisfy the requirements in the above theorem. If we neglect the labeling, so that any subtrees are identical if they are structurally identical, then $a_m$ satisfies the following recurrence:
\begin{equation}
    a_m = \sum_{\substack{0 \leq k_1,\ldots,k_d \leq d \\ \sum k_j = m-1}} \prod_{i=1}^d a_{k_i}
\end{equation}
Each summand represents a configuration in which the root node has a tree of size $k_1$ on its first child, size $k_2$ on its second child, and so on. We now need to adjust this recurrence to account for labeling. 

For some configuration $(k_1,k_2,\ldots,k_d)$, observe that we can combine the configurations as follows. There are $m$ slots that need to be labeled. The first must be the root. Then, we can place the labels from the first subtree in any of $\binom{m-1}{k_1}$ positions (while maintaining relative ordering within the subtree\footnote{This maintains the property that the label of any given node is less than that of all its descendants.}). Next, we can place the labels from the second subtree in any of $\binom{m-1-k_1}{k_2}$ positions, and so on. In total, there are $\frac{(m-1)!}{k_1! (m-1-k_1)!} \frac{(m-1-k_1)!}{(m-1-k_1-k_2)! k_2!} \ldots = \frac{(m-1)!}{k_1! k_2! \ldots k_d!}$ ways of placing the labels. This is simply the multinomial coefficient $\binom{m-1}{k_1,\ldots,k_d}$. Therefore, the correct recurrence relation is:
\begin{equation}
    a_m = \sum_{\substack{0 \leq k_1,\ldots,k_d \leq d\\ \sum k_j = m-1}} \frac{(m-1)!}{k_1! \ldots k_d!} \prod_{i=1}^d a_{k_i}.
\end{equation}
This can be written in a more convenient form if we define a $b_m \equiv \frac{a_m}{m!}$, so that we have:
\begin{equation}
    b_m = \frac{1}{m} \sum_{\substack{0 \leq k_1,\ldots,k_d \leq d \\ \sum k_j = m-1}} \prod_{i=1}^d b_{k_i}.
\end{equation}

To obtain a closed form for $b_m$, we define the generating function $G(z)=b_0 + b_1 z + b_2 z^2 + \ldots=\sum_m b_m z^m$.
\begin{align*}
    \qty(G(z))^d &= \sum_{m=0}^\infty \qty(\sum_{\substack{0 \leq k_1,\ldots,k_d \leq d \\ \sum k_j = m}} \prod_{i=1}^d b_{k_i}) z^m \\
    &= \sum_{m=0}^\infty (m+1) b_{m+1} z^{m} \\
    &= G'(z)
    \intertext{This is a first order equation that can be solved in closed form.}
    1 &= G'(z) G^{-d}(z) \\
    z+C &= \int G^{-d} \dd{G} \\
    z+C &= \frac{G(z)^{1-d}}{1-d} \\
    G(z) &= \sqrt[1-d]{-(d-1)(z+C)}
    \intertext{Since $G(0)=b_0=a_0=1$:}
    &= \sqrt[1-d]{-(d-1)\qty(z-\frac{1}{d-1})}
    \intertext{For brevity, write $\alpha=d-1$.}
    &= \sqrt[-\alpha]{-\alpha \qty(z-\alpha^{-1})}
\end{align*}
Now, observe that $\eval\dv[m]{G}{z}|_{z=0}=b_m m!=a_m$. We show by induction that
\begin{equation}
    \dv[m]{G}{z}= \frac{\prod_{n=1}^{m-1} (n \alpha+1)}{ (1-\alpha z)^{1+1/\alpha}}
\end{equation}
For the base case $m=1$, the derivative is $\frac{1}{(1-\alpha z)^{1+1/\alpha}}$. Then, assuming the induction hypothesis,
\begin{align*}
    \dv[m+1]{G}{z} &= \dv{z} \frac{\prod_{n=1}^{m-1} (n \alpha+1)}{ (1-\alpha z)^{1+1/\alpha}} =  \frac{\prod_{n=1}^{m} (n \alpha+1)}{ (1-\alpha z)^{1+1/\alpha}}
\end{align*}
This completes the induction. Finally, we see that:
\begin{align*}
    a_m &= \prod_{n=1}^{m-1} (n\alpha+1) \\
    &\leq \alpha^{m-1} \prod_{n=1}^{m-1} (n+1) \\
    &= (d-1)^{m-1} m!
\end{align*}
\end{proof}

\begin{corollary}\label{cor:counting}
We define a non-vanishing tuple $S$ of size $m$ as a tuple where
\begin{equation*}
    \qty[P_{s_1},\ldots,\qty[P_{s_m},P]] \neq 0.
\end{equation*}
There are at most $\mathscr{D}^{m} (m+1)!$ non-vanishing tuples of size $m$.
\end{corollary}
\begin{proof}
We argue that non-vanishing tuples are in one-to-one correspondence with the subtrees of size $m+1$ described in the previous theorem. We do this by construction.
\begin{algorithm}[H]
\caption{Finding all possible non-vanishing tuples of size $m$}\label{alg:counting}
\begin{algorithmic}[1]
\Procedure {NonvanishingTuples}{$S$, $F$, $m$}
\If{$m=0$}
    \State \textbf{return} $[P]$
\Else
    \State $R \gets []$
    \For{$v \in F$}
        \State $S[m] \gets v$
        \State $F' \gets F$ \Comment{Copy $F$ to $F'$}
        \State Remove $v$ from $F'$
        \State Add the children of $v$ to $F'$
        \State $R.\texttt{extend}(\textsc{NonvanishingTuples}(S, F', m-1))$
    \EndFor
    \State \textbf{return} $R$
\EndIf
\EndProcedure
\State $S_0 \gets \text{empty list of size } m+1$
\State $S_0[m+1] \gets P$
\State $F_0 \gets$ the children of $P$ in the support tree of $\mathcal{T}_P$
\State $\qty{S} \gets \textsc{NonvanishingTuples}(S_0, F_0, m)$
\end{algorithmic}
\end{algorithm}
In the above, $S$ represents a given non-vanishing tuple, and $F$ represents a frontier of Paulis whose support intersects with the support of at least one Pauli in $S$. This algorithm builds non-vanishing tuples from inside-out, first selecting all $P_{s_m}$ such that $\supp P_{s_m} \cap \supp P \neq \varnothing$, then finds $P_{s_{m-1}}$ such that $\supp P_{s_{m-1}} \cap \qty(\supp P_{s_m} \cup \supp P) \neq \varnothing$, and so on.

Now, we observe that this algorithm is constructing, structurally speaking, rooted subtrees of size $m+1$ in $\mathcal{T}_P$ (where the $+1$ is due to the root $P$). The recursion takes the current subtree, searches its `frontier' (i.e., the set of nodes that is connected to at least one node in the current subtree) and adds one of these frontier nodes to the current subtree, repeating until the subtree is size $m+1$. However, the algorithm does not only output structures: there is an ordering to each structure depending on the order in which node is visited. A given subtree may be counted several times, with different ordering. This ordering corresponds to a labeling of the nodes in the subtree with $\qty{1,\ldots,m}$, with the labels representing the order in which the nodes are visited. The constraint on the labeling is that no descendent of a node can be visited before its parent -- this simply means the label of any node must be strictly less than that of its descendants. Therefore, \cref{alg:counting} counts all labelled subtrees of size $m$ with the property that the label of any node is less than that of its descendants. Since $\mathcal{T}_P$ is at most $(\mathscr{D}+1)$-regular, by \cref{thm:counting}, there are at most $\mathscr{D}^{m} (m+1)!$ non-vanishing tuples.
\end{proof}
\begin{remark}\label{rem:overcount}
It is not guaranteed that every tuple returned by \cref{alg:counting} will be nonvanishing, since it is possible, for example, that $P_{s_1}$ commutes with $\qty[P_{s_2},\ldots,[P_{s_m},P]]$. However, it is guaranteed that every non-vanishing tuple is contained in the set $\qty{S}$ found by \cref{alg:counting}.
\end{remark}
\begin{remark}
\cref{cor:counting} is a similar approach to \citep{tang2021}. However, they use the known result that the number of structurally unique subtrees with size $m$ of a $d$-regular tree is $\frac{1}{(d-2)m+1} \binom{(d-1)m}{m}$ \citep{alon1991,knuth1969}. This scales with $\sim (e \cdot d)^m$, and if we want to count all non-vanishing tuples, to allow for relabelling, we are forced to introduce a factor $m!$. This gives a bound on the number of non-vanishing tuples $\sim (e \cdot d)^m m!$. By being more careful with relabeling (since not all labelings are permitted), the counting in \cref{thm:counting} is better by a factor $\order{\frac{e^m}{m}}$.
\end{remark}

\begin{theorem}\label{lem:norm-bound2}
The spectral norm of the iterated commutator is bounded by
\begin{equation}
    \norm{\qty[H^m P]} \leq (2 \mathscr{D} \norm{\Theta}_\infty)^m (m+1)! \label{eqn:time-deriv-bound2}
\end{equation}
where the $\ell_\infty$ norm denotes $\norm{\Theta}_\infty = \max_i \abs{\theta_i}$.
\end{theorem}
\begin{proof}
We refer to the expression for the iterated commutator in \cref{eqn:poly-expand}:
\begin{equation*}
    \qty[H^m P] = \sum_{S \in \qty{1, \ldots, r}^m} \Theta^{\alpha(S)} \qty[P_{s_1}, [P_{s_2}, \ldots, [P_{s_m}, P]]]
\end{equation*}
First, each commutator in this sum has norm at most $2^m$, since $\norm{\comm{A}{B}} \leq 2\norm{A} \norm{B}$ and all Paulis have spectral norm $\norm{P}=1$. Also,
\begin{equation*}
    \Theta^{\alpha(S)}  = \prod_{i=1}^n \theta_i^{\alpha(S)_i} \leq \prod_{i=1}^n \norm{\Theta}_\infty^{\alpha(S)_i} = \norm{\Theta}_\infty^{\abs{\alpha(S)}} = \norm{\Theta}_\infty^m
\end{equation*}
Finally, by \cref{cor:counting}, there are at most $\mathscr{D}^m (m+1)!$ non-vanishing term tuples. Combining these together, we get that
\begin{align*}
    \norm{\qty[H^m P]} &\leq \sum_{S \in \qty{1, \ldots, r}^m} \Theta^{\alpha(S)} \norm{\qty[P_{s_1}, \qty[P_{s_2}, \ldots, \qty[P_{s_m}, P]]]} \numberthis \label{eqn:triangle-ineq} \\
    &\leq \underbrace{\norm{\Theta}_\infty^m}_{\text{From $\Theta^{\alpha(S)}$}} \quad \underbrace{2^m}_{\text{From each commutator}} \quad \underbrace{\mathscr{D}^m (m+1)!}_{\text{Counting non-vanishing $S$}} \\
    &= (2 \mathscr{D} \norm{\Theta}_\infty)^m (m+1)!
\end{align*}
\end{proof}
\begin{theorem}\label{lem:commute-bound2}
For commuting Hamiltonians (i.e., every term $P_i$ in the Hamiltonian commutes with every other term), when $P$ is a single-qubit observable
\begin{equation}
    \norm{\qty[H^m P]} \leq (2(\mathscr{D}+1)\norm{\Theta}_\infty)^m \label{eqn:commute-bound2}
\end{equation}
for all $m$. 
\end{theorem}
\begin{proof}
We separate $H=H_1+H_2$, where $H_1$ is composed of all terms in the Hamiltonian that have a support that overlaps with the support of $P$. We then inductively show \cref{eqn:commute-bound2} by proving the strong statement
\begin{equation*}
    \qty[H^m P] = \qty[H_1^{m} P].
\end{equation*}
Assuming the induction hypothesis, we have:
\begin{align*}
    \qty[H^{m+1} P] &= \qty[H_1^{m+1} P] + \comm{H_2}{\qty[H_1^m P]}
    \intertext{Using the identity $\qty[A^m B] = \sum_{k=0}^m (-1)^k \binom{m}{k} A^{m-k} B A^k$:}
    &= \qty[H_1^{m+1} P] + \sum_{k=0}^m (-1)^k \binom{m}{k} \qty(H_2 H_1^{m-k} P H_1^k - H_1^{m-k} P  H_1^k H_2)
    \intertext{By commutativity of the Hamiltonian $\comm{H_1}{H_2}=0$, so we can freely rearrange $H_2$ and $H_1$. That is, $H_2 H_1^{m-k} P H_1^k=H_1^{m-k} H_2 P H_1^k$ and similarly $H_1^{m-k} P  H_1^k H_2=H_1^{m-k} P H_2 H_1^k$. This gives:}
    &= \qty[H_1^{m+1} P] + \sum_{k=0}^m (-1)^k \binom{m}{k} H_1^{m-k} \comm{H_2}{P} H_1^k
    \intertext{By definition of $H_2$, $\supp H_2 \cap \supp P = \varnothing$, so they commute.}
    &= \qty[H_1^{m+1} P]
\end{align*}
Since $H_1$ contains at most $\mathscr{D}+1$ terms, by applying the triangle inequality for the matrix norm, we have $\norm{H_1} \leq \norm{\Theta}_\infty (\mathscr{D}+1)$. Finally, applying $\norm{\qty[A^m B]} \leq 2^m \norm{A}^m \norm{B}$, we find $\norm{\qty[H^m P]} \leq (2(\mathscr{D}+1) \norm{\Theta}_\infty)^m$, as desired.

Notably, \cref{eqn:commute-bound2} is smaller than the general bound in \cref{eqn:time-deriv-bound2} by a factor $\sim m!$. This can be attributed to the fact that when $H$ is commuting, the support $\qty[H^m P]$ never grows beyond a ring around $P$ -- more precisely, $\supp \qty[H^m P] \subseteq \supp H_1$ for all $m$.
\end{proof}

\section{Proof of \texorpdfstring{\cref{thm:err-scale}}{Theorem \ref{thm:err-scale}}}\label{sec:proof-err}
We first establish one preliminary: since we use Chebyshev regression, it is natural to first write the expectation in the form $\expval{P(t)} = \sum_{m=0}^\infty b_m T_m(2t/A - 1)$, where $T_m$ is the $m$th Chebyshev polynomial. However, we need to find an expression relating $b_m$ to the Taylor expansion coefficients in \cref{eqn:taylor-coeff} -- specifically, since we are interested in expressing $c_1$ as a function of $\qty{b_m}$.
\begin{lemma}\label{lem:linear}
If we have a polynomial of degree $L$ represented in the Chebyshev basis $$f(t) = \sum_{m=0}^L b_m T_m(2t/A - 1),$$ if we write the same polynomial with $f(t) = \sum_{m=0}^L c_m t^m$, we have
\begin{equation}
    c_1 = -\frac{2}{A} \sum_{m=1}^L (-1)^m b_m m^2 \label{eqn:first-order-coeff}
\end{equation}
\end{lemma}
\begin{proof}
Since $c_1 \equiv f'(0)$:
\begin{align*}
    c_1 &= \sum_{m=1}^L b_m \dv{t} \eval{T_m(2t/A-1)}_{t=0} \\
    &= \sum_{m=1}^L \frac{2}{A} b_m T_m'(-1)
    \intertext{We apply the formula $T_m'(-1) = (-1)^{m+1} m^2$ \citep[Section 2.4.5]{mason2003}.}
    &= -\frac{2}{A} \sum_{m=1}^L (-1)^m b_m m^2.
\end{align*}
\end{proof}

\begin{theorem}[Error bound]\label{thm:master-error2}
The estimator $\tilde{c}_1$ proposed in \cref{alg:first-comm} achieves an error:
\begin{equation}
    \mathbb{E}\qty[\frac{\qty(c_1 - \tilde{c}_1)^2}{\gamma^2}] \leq \frac{8}{(A/\tau)^2} \qty[\quad \underbrace{\ \frac{(L-\frac{1}{2})^4}{5} \sigma^2 \ }_{\text{Noise induced}} \quad + \underbrace{4L^4 (L+1)^2 (A/4\tau)^{2L}}_{\text{Modeling error}}] \label{eqn:master-error2}.
\end{equation}
Since $\gamma$ defines a typical scale for the Hamiltonian coefficients, $\sqrt{\mathbb{E}\qty[\qty(c_1 - \tilde{c}_1)^2/\gamma^2]}$ can be interpreted as a relative error.
\end{theorem}
\begin{proof}
We use the identity $\mathbb{E}[(c_1-\tilde{c}_1)^2] = (c_1 - \mathbb{E}[c_1])^2 + \mathbb{V}[\tilde{c}_1]$, where $\mathbb{V}[\tilde{c}_1]$ is variance in the estimator $\tilde{c}_1$ due to randomness in the dataset.
\begin{align*}
    \mathbb{V}[\tilde{c}_1] &= \mathbb{V}\qty[\frac{2}{A} \sum_{m=1}^{L-1} (-1)^m \tilde{b}_m m^2]
    \intertext{Applying \cref{eqn:cheb-coeff}:}
    &= \mathbb{V}\qty[\frac{2}{A} \frac{2}{L} \sum_{m=1}^{L-1} (-1)^m m^2 \sum_{\ell=1}^{L} y_\ell T_m(z_\ell)]
    \intertext{Swapping the order of the two sums:}
    &= \mathbb{V}\qty[\frac{2}{A} \frac{2}{L} \sum_{\ell=1}^{L} y_\ell \sum_{m=1}^{L-1} (-1)^m m^2 T_m(z_\ell)]
    \intertext{Since each of the $y_\ell$ are statistically independent:}
    &= \frac{16}{L^2 A^2}\sum_{\ell=1}^{L} \sigma_\ell^2 \qty(\sum_{m=1}^{L-1} (-1)^m m^2 T_m(z_\ell))^2 \numberthis \label{eqn:var-est} \\
    &\leq \frac{16\sigma^2}{L^2 A^2} \sum_{\ell=1}^L \qty(\sum_{m=1}^{L-1} (-1)^m m^2 T_m(z_\ell))^2 \qq{where $\sigma^2 = \max_\ell \sigma_\ell^2$} \\
    &= \frac{16\sigma^2}{L^2 A^2}\sum_{\ell=1}^L \sum_{m_1=1}^{L-1} \sum_{m_2=1}^{L-1} (-1)^{m_1+m_2} (m_1m_2)^2 T_{m_1}(z_\ell) T_{m_2}(z_\ell) \\
    &= \frac{16\sigma^2}{L^2 A^2}\sum_{m_1=1}^{L-1} \sum_{m_2=1}^{L-1} (-1)^{m_1+m_2} (m_1m_2)^2 \sum_{\ell=1}^L T_{m_1}(z_\ell) T_{m_2}(z_\ell)
    \intertext{By the discrete orthogonality conditions, the sum is non-vanishing only when $m_1=m_2$.}
    &= \frac{16\sigma^2}{L^2 A^2}\frac{L}{2} \sum_{m=1}^{L-1} m^4
    \intertext{Since $\sum_{k=1}^n k^4 = \frac{n(n+1)(2n+1)(3n^2+3n-1)}{30}$:}
    &= \frac{4\sigma^2}{L^2 A^2}\frac{L (L-1) L (2L-1) (3L^2 - 3L - 1)}{15} \\
    &= \frac{4(L-1)(2L-1)(3L^2-3L-1)}{15A^2} \sigma^2 \\
    &\leq \frac{8\qty(L-\frac{1}{2})^4}{5A^2} \sigma^2 .
\end{align*}
Next, we evaluate the bias $(c_1 - \mathbb{E}[\tilde{c}_1])^2$. Since $\mathbb{E}[\tilde{c}_1]$ corresponds to Chebyshev interpolation with no noise, we make use of a theorem \citep[Equation 4.2]{howell1991} concerning the derivative error bounds for Chebyshev interpolation. This theorem says that if $\tilde{f}$ is a degree $L-1$ Chebyshev interpolation of some function $f$, $\abs{f'(0)-\tilde{f}'(0)} \leq \frac{\abs{\omega_1(0)} \abs{f^{(L)}}}{L!}$, where $f^{(L)}$ is the $L$th derivative of $f$, $\omega_1(t) \equiv \prod_{\ell=1}^{L-1} (t-\eta_\ell)$, $t_{\ell} \leq \eta_\ell \leq t_{\ell+1}$ and $\abs{f} \equiv \sup_{0 \leq t \leq A} \abs{f(t)}$. Applied to our case, $f = \expval{P(t)}$, so $f^{(L)}(t) = i^m \Tr(\qty[H^m P] \rho_0(t)) \implies \norm{f^{(L)}(t)}\leq\norm{\qty[H^m P]}$. Then:
\begin{align*}
    \abs{c_1 - \mathbb{E}[\tilde{c}_1]} &\leq \abs{f' -\tilde{f}'}  \\
    &\leq \frac{\abs{\omega_1}}{L!} \abs{f^{(L)}} \numberthis \label{eqn:model-err-bound}
    \intertext{Applying \cref{lem:norm-bound}:}
    &\leq \frac{\abs{\omega_1}}{L!} \gamma^L (L+1)!
    \intertext{It remains to upper bound $\abs{\omega_1(0)}$.}
    \abs{\omega_1} &\leq \abs{\prod_{\ell=1}^{L-1} \frac{A}{2} \qty(-1-z_{\ell+1})} \\
    &= \qty(A/2)^{L-1} \prod_{\ell=2}^{L} \qty(-1-z_{\ell}) 
    \intertext{Note that $\prod_{\ell=2}^{L} \qty(z-z_{\ell}) = \frac{T_L(z)}{2^{L-1} (z-z_1)}$, since this is the unique monic polynomial with roots at $z_2, \ldots, z_L$. Then:}
    &= (A/2)^{L-1}\abs{\frac{\cos(L \arccos(1))}{2^{L-1}(-1-z_1)}}
    \intertext{We use the fact that $\cos x < 1-x^2/2.3$ for $\abs{x}<1$, so that $\abs{\frac{1}{-1-z_1}}=\frac{1}{1-\cos(\pi/2L)} \leq \frac{1}{(\pi/2L)^2/2.3} \leq L^2$ for $L \geq 2$, so:}
    &\leq L^2 (A/4)^{L-1}
    \intertext{Therefore,}
    \abs{c_1 - \mathbb{E}[\tilde{c}_1]} &\leq \frac{4 L^2 (L+1) (A \gamma / 4)^L}{A}
\end{align*}
In summary, the total error is:
\begin{align*}
    \mathbb{E}[\qty(c_1 - \tilde{c}_1)^2] &\leq \frac{8\qty(L-\frac{1}{2})^4\sigma^2}{5A^2} + \frac{16L^4 (L+1)^2 (A\gamma/4)^{2L}}{A^2}, \\
    \mathbb{E}\qty[\frac{\qty(c_1 - \tilde{c}_1)^2}{\gamma^2}] &\leq \frac{1}{(A \gamma)^2} \qty[\frac{8(L-\frac{1}{2})^4}{5} \sigma^2 + 16L^4(L+1)^2(A \gamma/4)^{2L}].
\end{align*}
\end{proof}

\begin{remark}
The above error bound can be written as:
\begin{equation}
    \mathbb{E}\qty[\frac{(c_1-\tilde{c}_1)^2}{\gamma^2}] \leq \order{\frac{\boldsymbol{L^4}}{(A/\tau)^2} \qty[\sigma^2 + \boldsymbol{L^2 (A/\tau)^{2L}}]}.
\end{equation}
By viewing the estimator $\tilde{c}_1$ as a generalized finite-difference estimator $f'(0) \approx \frac{f(\epsilon)-f(0)}{\epsilon}$, we argue that the terms in this expression (with the exception of those in bold) are fundamental:
\begin{itemize}
    \item The inverse dependence on evolution time $A$ corresponds to dependence on $1/(\Delta t)^2$ in the finite-difference estimator.
    \item The $\sigma^2$ dependence originates in the noisiness of measurements we take.
\end{itemize}
On the other hand, the terms in bold are not fundamental. More precisely, they originate in the fact that we can only evolve our system forward in time. If it were possible to evolve backward in time, we would have access to the central difference estimator $f'(0) \approx \frac{f(\epsilon)-f(-\epsilon)}{2\epsilon}$. This would improve our estimator by reducing the terms in bold.
\begin{itemize}
    \item Currently, the dependence on $(L+1)(A/\tau)^{L}$ measures the modeling error. This is analogous to the finite difference estimator $f'(x) \approx \frac{f(x+\epsilon)-f(x)}{\epsilon}$ having an error that scales with $\order{\epsilon \frac{f''(\xi)}{2!}}$, where $x \leq \xi \leq x+\epsilon$. However, with a central difference method has an error that scales with error $\order{\epsilon^2 \frac{f'''(\xi)}{3!}}$. Therefore, roughly speaking, we would expect the modeling error to improve by at least a factor $\frac{A/\tau}{L+1}$ if we had access to backwards time evolution.
    \item The $L^4$ dependence changes to a $L^2$ dependence. This is because, by placing $t=0$ in the center of the Chebyshev roots $\qty{z_i}$ (rather than at the extreme end $z=-1$), the expression for $c_1$ changes its form from $\sim \sum_m b_m m^2$ to $\sim \sum_m b_m m$. Plugging this into the derivation for the variance error, we get scaling with $L^2$.
\end{itemize}
\end{remark}

\begin{theorem}[Query complexity for one coefficient]\label{thm:err-scale2}
Fix some failure probability $\delta$ and an error $\epsilon$. Assume that we have access to an unbiased (single-shot) estimator of $\expval{P(t)}$ with variance $\sigma^2 \leq 1$. Then there is some choice of $A \sim \tfrac{1}{\gamma}$ and $L \sim \log \epsilon^{-1}$ such that with
\begin{equation}
    \order{\log(1/\delta) \polylog(1/\epsilon) \epsilon^{-2}}
\end{equation}
query complexity, we can construct an estimator $\tilde{c}_1$ such that $\frac{\abs{c_1-\tilde{c}_1}}{\gamma} \leq \epsilon$, except with failure probability at most $\delta$.
\end{theorem}
\begin{proof}
To ensure $\frac{\abs{c_1-\tilde{c}_1}}{\gamma} \leq \epsilon$, it suffices to guarantee $\frac{\abs{\tilde{c}_1 - \mathbb{E}[\tilde{c}_1]}}{\gamma} + \frac{\abs{\mathbb{E}[\tilde{c}_1]-c_1}}{\gamma} \leq \epsilon$. Since the bias $\abs{\mathbb{E}[\tilde{c}_1]-c_1}$ was shown to be upper bounded by $\frac{4L^2(L+1) (A\gamma/4)^L}{A}$ (see \cref{sec:proof-err}) we require:
\begin{align*}
    \frac{\abs{\tilde{c}_1 - \mathbb{E}[\tilde{c}_1]}}{\gamma} &\leq \epsilon - \frac{4L^2(L+1) (A\gamma/4)^L}{(A\gamma)}
\end{align*}
with failure probability at most $\delta$. We demonstrate that a median-of-means estimator \citep{jerrum1986,nemirovskii1983}, with
\begin{equation}
    K = 2 \log(2/\delta) \qq{and} N = \order{\polylog(1/\epsilon)\epsilon^{-2}}
\end{equation}
satisfies this. That is, we will take the median of $K$ independent sample means, where each sample mean is over $N$ estimates of $\tilde{c}_1$. 

It suffices to have $N \leq 34 \sigma_1^2 \qty(\epsilon - \frac{4L^2(L+1) (A\gamma/4)^L}{(A\gamma)})^{-2}$, where $\sigma_1^2$ is the variance of the $\tilde{c}_1$ estimator with a single query for each point in the dataset. As found in the proof of \cref{thm:master-error2}, $\sigma_1^2 \leq \frac{8L^4}{5(A\gamma)^2} \sigma^2$. Thus, the above upper bound on $N$ becomes:
\begin{align*}
    N &\leq \frac{34\sigma^2 \cdot 8L^4/5(A\gamma)^2}{\qty(\epsilon - \frac{4L^2(L+1)(A\gamma/4)^L}{(A\gamma)})^2} \numberthis \label{eqn:n-bound} \\
    &\leq 68\sigma^2 \qty(\frac{A\gamma}{L^2} \epsilon - 4(L+1)(A\gamma/4)^L)^{-2} 
    \intertext{Let us choose $L \sim \Theta(\log \epsilon^{-1})$. That is, there exists some $0 < L_0 < L_1$ such that $L_0 \log \epsilon^{-1} \lesssim L \lesssim L_1 \log \epsilon^{-1}$, where $a(x) \lesssim f(x)$ denotes $a(\epsilon) < f(\epsilon)$ for $\epsilon^{-1}$ greater than some fixed $\epsilon_0^{-1}$. Then $(A\gamma/4)^L \lesssim (A \gamma/4)^{L_1 \log \epsilon^{-1}} = \epsilon^{L_1 \log(4/A \gamma)}$. Then, setting $A=4e^{-3/L_1}/\gamma$, we have $(A\gamma/4)^L \lesssim \epsilon^3$. Then:}
    &= \order{\qty(\frac{e^{-3L_1}}{\qty(L_1\log\epsilon^{-1})^2} \epsilon - L_1 \qty(\log \epsilon^{-1}) \epsilon^3)^{-2}} \\
    &= \order{\polylog(\epsilon^{-1})\qty(\epsilon - \qty(\epsilon \log \epsilon^{-1})^3)^{-2}} \\
    \intertext{Since $\log \epsilon^{-1} = \order{\epsilon^{-m}}$ for every positive exponent $m$:}
    &= \order{\polylog(\epsilon^{-1}) \epsilon^{-2}}
\end{align*}
The query complexity is $NK L$, since $L$ queries are required for a single-shot estimate of $\tilde{c}_1$ (from evaluating the expectation value of the observable at $L$ different evolution times). However, since $L \sim \Theta(\log \epsilon^{-1})$, we still have $NKL = \order{\log(1/\delta) \polylog(1/\epsilon)\epsilon^{-2}}$.
\end{proof}

\section{Proof of \texorpdfstring{\cref{thm:master-theorem}}{Theorem \ref{thm:master-theorem}}}\label{sec:master-theorem}
The following is a sequence of results used when efficiently parallelizing our measurements for Hamiltonian learning with unitary dynamics.
\begin{lemma}[Term selection]\label{thm:param-select2}
Let $P$ be some Pauli operator such that there exists some $i \in \qty{1, \ldots, r}$ where $\supp P \subseteq \supp P_i$ and $\frac{i\comm{P_i}{P}}{2} \neq 0$. Let $X = \supp P_i$. Let
\begin{gather}
    X = \supp P_i ,\\
    Y = \qty(\bigcup \qty{\supp P_j \mid \supp P_j \cap X \neq \varnothing}) \setminus X, \\
    Z = (X \cup Y)', \\
    \rho_0 = \qty(\frac{\mathbb{I} + i\comm{P_i}{P}/2}{2^{\abs{X}}})^{(X)} \otimes \qty(\frac{\mathbb{I}}{2^{\abs{Y}}})^{(Y)} \otimes \rho_0^{(Z)} \label{eqn:state}.
\end{gather}
In words, $Y$ is a neighborhood around $X$ that contains the support of all Paulis that intersect with $X$, and $Z$ is the set of all qubits that are not in $X \cup Y$. The state $\rho_0$ is defined such that for all qubits in $Y$, it is the maximally mixed state and for qubits inside $X$, $\rho_0$ is defined in a way such that $\Tr(i\comm{P_i}{P} \rho_0^{(X)}/2)=1$, and for all other qubits, $\rho_0$ can be anything. Then:
\begin{equation}
    \Tr(i \comm{H}{P} \rho_0) = \theta_i .
\end{equation}
\end{lemma}
\begin{proof}

It suffices to show that for all $j$ such that $j \neq i$, $\Tr(i \comm{P_j}{P} \rho_0)=0$. This is trivially true in the case where $\supp P_j \cap \supp P = \varnothing$, since the commutator vanishes. There are two remaining cases.
\begin{enumerate}[label=\textsc{Case \arabic*.},left=0pt]
    \item $P_j$ acts nontrivially on some set of qubits in $Y$. Note that $P^{(Y)}= \mathbb{I}$. Then $\comm{P_j}{P} = \comm{P_j^{(X)}}{P^{(X)}} \otimes P_j^{(Y)} \otimes \mathbb{I}^{(Z)}$, and:
    \begin{align*}
        \Tr(i \comm{P_j}{P} \rho_0) &= \Tr(i \comm{P_j^{(X)}}{P^{(X)}} \rho_0^{(X)}) \Tr(P_j^{(Y)} \rho_0^{(Y)}) \Tr(\mathbb{I}^{(Z)} \rho_0^{(Z)})
        \intertext{Since we assumed $P_j^{(Y)}$ acts nontrivially on at least one qubit in $Y$, and $\rho_0^{(Y)} \propto \mathbb{I}$, $\Tr(P_j^{(Y)} \rho_0^{(Y)})=0$.}
        &= 0
    \end{align*}
    \item $P_j$ acts trivially on all qubits in $Y$. Then, we have $i \comm{P_j}{P} = i\comm{P_j^{(X)}}{P^{(X)}} \otimes \mathbb{I}^{(Y)} \otimes \mathbb{I}^{(Z)}$.
    \begin{align*}
        \Tr(i \comm{P_j}{P} \rho_0) &= \Tr(i \comm{P_j^{(X)}}{P^{(X)}} \frac{\mathbb{I} + i\comm{P_i}{P}/2}{2^{\abs{X}}})
        \intertext{Since $\comm{P_j^{(X)}}{P^{(X)}}$ is traceless:}
        &= \frac{\Tr(i \comm{P_j^{(X)}}{P^{(X)}} i\comm{P_i^{(X)}}{P^{(X)}}/2)}{2^{\abs{X}}}
        \intertext{We now observe that $P_i^{(X)} \neq P_j^{(X)} \implies \comm{P_i^{(X)}}{P^{(X)}} \neq \comm{P_j^{(X)}}{P^{(X)}}$. Since $\comm{P_i^{(X)}}{P^{(X)}}$ and $\comm{P_j^{(X)}}{P^{(X)}}$ are both proportional to different Pauli matrices, by the orthonormality of Pauli matrices, the trace vanishes.}
        &= 0
    \end{align*}
\end{enumerate}
Finally:
\begin{align*}
    \Tr(i \comm{H}{P} \rho_0) &= \sum_{j=1}^n \theta_j \Tr(i \comm{P_j}{P} \rho_0) \\
    &= \theta_i \Tr(i \comm{P_i}{P} \rho_0) \\
    &= \theta_i \qq{by definition of $\rho_0$}
\end{align*}
\end{proof}

\begin{lemma}[Simultaneous inference for a partition]\label{thm:simul-inf2}
Let $\mathbf{V}_i$ be a partition in a coloring of $\mathcal{G}^2$. The coefficient for each Pauli in $\mathbf{V}_i$ can be inferred with up to an error $\epsilon \norm{\Theta}_\infty$, with failure probability for each \textbf{individual} coefficient being at most $\delta$ (so the overall failure probability is upper bounded by $\delta \abs{\mathbf{V}_i}$). This can be done with query complexity
\begin{equation}
    \order{\mathscr{D}^2 \log(1/\delta) \polylog(\mathscr{D}/\epsilon) \epsilon^{-2}}. \label{eqn:partition-complexity2}
\end{equation}
\end{lemma}
\begin{proof}
For each $P_{i,j} \in \mathbf{V}_i$, let $P_{i,j}'$ be a single-qubit Pauli such that $\comm{P_{i,j}}{P_{i,j}'} \neq 0$. Let $M=(\supp \mathbf{V}_i)'$. Let
\begin{equation}
    \rho_0 = \qty(\bigotimes_{j; P_{i,j} \in \mathbf{V}_i} \qty(\frac{\mathbb{I} + P_{i,j}'}{2})) \otimes \qty(\frac{\mathbb{I}}{2^{\abs{M}}})^{(M)}
\end{equation}
Here, $M$ is the ``moat". Now, for any $P_{i,j}$, there is a straightforward labeling of qubits that maps $\rho_0$ onto the structure $\rho_0^{(X)} \otimes \rho_0^{(Y)} \otimes \rho_0^{(Z)}$ in \cref{thm:param-select}.
\begin{gather}
    X = \supp(P_j) \\
    Y = M \\
    Z = \supp \mathbf{V}_i \setminus \supp P_{i,j}
\end{gather}
By \cref{alg:first-comm,thm:param-select}, so long as we are able to evaluate expectation values $\Tr(P_{i,j}' \rho_0(t))$, we can find $\Tr(i\comm{H}{P_{i,j}'} \rho_0)$, hence we can also find $\qty{\theta_{i,j} \mid P_{i,j} \in \mathbf{V}_i}$. Therefore, it suffices to show that (for some fixed time $t$) the expectation values $\Tr(P_{i,j}' \rho_0(t))$ can be evaluated for all $j$. This is not difficult, since $\{ P_{i,j}'\}$ is a set of Pauli operators with non-overlapping support, so they commute and can be measured simultaneously. 

Since the error in \cref{thm:err-scale} was relative to $\gamma \equiv 2 \mathscr{D} \norm{\Theta}_\infty$, after rescaling $\epsilon \rightarrow \frac{\epsilon}{2 \mathscr{D}}$ we find that we require $L \sim \log \epsilon^{-1}$ groups of $N \sim \mathscr{D}^2\polylog(\mathscr{D}/\epsilon) \epsilon^{-2}$ measurements, with $K \sim \log(1/\delta)$ (where we take the mean within each group, and the median across all the means). The total query complexity is $NLK$, which gives the complexity in \cref{eqn:partition-complexity2}.
\end{proof}

Finally, since our parallelization technique relies on a graph coloring, we provide a formal definition below.
\begin{definition}[Graph coloring]\label{defn:color}
A $C$-coloring of a graph is a labelling of the vertices in a graph with exactly $C$ colors such that no two vertices with the same color share an edge. The minimum number of colors required to color a graph $\mathcal{G}$ is known as its chromatic number $\chi(\mathcal{G})$. We will write a $C$-coloring as $\qty{\mathbf{V}_i \mid i = 1, \ldots, C}$, where $\mathbf{V}_i=\qty{P_{i,1}, P_{i,2}, \ldots}$ is called a partition, and is the set of Paulis that is colored with $i$. We will also define the support of a partition to be:
\begin{equation}
    \supp \mathbf{V}_i = \bigcup_{P_{i,j} \in \mathbf{V}_i} \supp(P_{i,j}).
\end{equation}
Furthermore, by Vizing's theorem \citep{vizing1965}, any graph with degree $\Delta$ can be colored with at most $\chi(\mathcal{G}) \leq \Delta+1$ colors. This coloring can be found in $\order{\Delta}$ time with a greedy algorithm \citep{mitchem1976}. Since the squared interaction graph has degree at most at most $\mathscr{D} \qty(\mathscr{D}-1) + 1 = \mathscr{D}^2 - \mathscr{D} + 1$, the graph can be colored with at most $\mathscr{D}^2 - \mathscr{D} + 2$ colors. For non-trivial Hamiltonians ($\mathscr{D} \geq 2$), this is at most $\mathscr{D}^2$ colors.
\end{definition}
\begin{example}\label{ex:sq-coloring}
The following is an example of a coloring of the squared interaction graph $\mathcal{G}^2$ for the 9-qubit TFIM model in \cref{ex:tfim-model}.
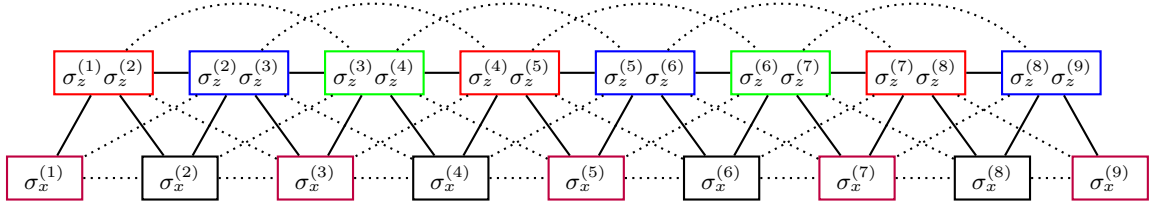
\begin{figure}[H]
    \centering
    \begin{tikzpicture}[node distance={18mm}, thick, main/.style = {draw, rectangle,minimum width=1cm}
    ] 
\node[main, draw=red] (0)               {$\sigma_z^{(1)} \sigma_z^{(2)}$}; 
\node[main, draw=blue] (1) [right of=0] {$\sigma_z^{(2)} \sigma_z^{(3)}$};
\node[main, draw=green] (2) [right of=1] {$\sigma_z^{(3)} \sigma_z^{(4)}$};
\node[main, draw=red] (3) [right of=2] {$\sigma_z^{(4)} \sigma_z^{(5)}$};
\node[main, draw=blue] (4) [right of=3] {$\sigma_z^{(5)} \sigma_z^{(6)}$};
\node[main, draw=green] (5) [right of=4] {$\sigma_z^{(6)} \sigma_z^{(7)}$};
\node[main, draw=red] (6) [right of=5] {$\sigma_z^{(7)} \sigma_z^{(8)}$};
\node[main, draw=blue] (7) [right of=6] {$\sigma_z^{(8)} \sigma_z^{(9)}$};

\node[main] (8) [draw=purple,  below left=8mm and -4mm of 0] {$\sigma_x^{(1)}$};
\node[main] (9) [draw=black, below left=8mm and -4mm of 1] {$\sigma_x^{(2)}$};
\node[main] (10) [draw=purple, below left=8mm and -4mm of 2] {$\sigma_x^{(3)}$};
\node[main] (11) [draw=black, below left=8mm and -4mm of 3] {$\sigma_x^{(4)}$};
\node[main] (12) [draw=purple, below left=8mm and -4mm of 4] {$\sigma_x^{(5)}$};
\node[main] (13) [draw=black, below left=8mm and -4mm of 5] {$\sigma_x^{(6)}$};
\node[main] (14) [draw=purple, below left=8mm and -4mm of 6] {$\sigma_x^{(7)}$};
\node[main] (15) [draw=black, below left=8mm and -4mm of 7] {$\sigma_x^{(8)}$};
\node[main] (16) [draw=purple, below right=8mm and -4mm of 7] {$\sigma_x^{(9)}$};

\draw[-] (0) -- (1);
\draw[-] (1) -- (2);
\draw[-] (2) -- (3);
\draw[-] (3) -- (4);
\draw[-] (4) -- (5);
\draw[-] (5) -- (6);
\draw[-] (6) -- (7);

\draw[-] (0) -- (8);
\draw[-] (0) -- (9);
\draw[-] (1) -- (9);
\draw[-] (1) -- (10);
\draw[-] (2) -- (10);
\draw[-] (2) -- (11);
\draw[-] (3) -- (11);
\draw[-] (3) -- (12);
\draw[-] (4) -- (12);
\draw[-] (4) -- (13);
\draw[-] (5) -- (13);
\draw[-] (5) -- (14);
\draw[-] (6) -- (14);
\draw[-] (6) -- (15);
\draw[-] (7) -- (15);
\draw[-] (7) -- (16);

\draw[dotted] (0) to[out=45,in=135] (2);
\draw[dotted] (1) to[out=45,in=135] (3);
\draw[dotted] (2) to[out=45,in=135] (4);
\draw[dotted] (3) to[out=45,in=135] (5);
\draw[dotted] (4) to[out=45,in=135] (6);
\draw[dotted] (5) to[out=45,in=135] (7);

\draw[dotted] (8) -- (9);
\draw[dotted] (9) -- (10);
\draw[dotted] (10) -- (11);
\draw[dotted] (11) -- (12);
\draw[dotted] (12) -- (13);
\draw[dotted] (13) -- (14);
\draw[dotted] (14) -- (15);
\draw[dotted] (15) -- (16);

\draw[dotted] (1) -- (8);
\draw[dotted] (2) -- (9);
\draw[dotted] (3) -- (10);
\draw[dotted] (4) -- (11);
\draw[dotted] (5) -- (12);
\draw[dotted] (6) -- (13);
\draw[dotted] (7) -- (14);

\draw[dotted] (0) -- (10);
\draw[dotted] (1) -- (11);
\draw[dotted] (2) -- (12);
\draw[dotted] (3) -- (13);
\draw[dotted] (4) -- (14);
\draw[dotted] (5) -- (15);
\draw[dotted] (6) -- (16);

\end{tikzpicture} 
    \caption[Coloring of a squared interaction graph $\mathcal{G}$]{A $5$-coloring of the squared interaction graph $\mathcal{G}^2$. Solid edges indicate the distance between the nodes is 1 in $\mathcal{G}$, and dashed edges indicate the distance in $\mathcal{G}$ is 2.}
    \label{fig:sq-coloring}
\end{figure}

\end{example}

\begin{theorem}[Hamiltonian learning with unitary dynamics]\label{thm:master-theorem2}
Fix a sparsely interacting Hamiltonian $H$ that has $r$ terms in its Pauli expansion with coefficients $\Theta$. For the appropriate choice of Chebyshev degree $L$ and evolution time $A$, \cref{alg:partition-learn} solves the quantum Hamiltonian learning problem (with an additive error $\epsilon \norm{\Theta}_\infty$ and failure probability at most $\delta$) with query complexity
\begin{equation}
    \order{\frac{\mathscr{D}^4 \log (r/\delta) \polylog(\mathscr{D}/\epsilon)}{\epsilon^{2}}}, \label{eqn:total-query2}
\end{equation}
and classical processing time complexity
\begin{equation}
    \order{\frac{\mathscr{D}^2 r \log (r/\delta) \polylog(\mathscr{D}/\epsilon)}{\epsilon^{2}}}. \label{eqn:total-time2}
\end{equation}
\end{theorem}
\begin{proof}
After finding a $\mathscr{D}^2$ coloring for the squared interaction graph $\mathcal{G}^2$, we use the result from \cref{thm:simul-inf} to simultaneously infer the coefficients for each partition in the graph coloring. For each partition $\mathbf{V}_i$, with reference to \cref{thm:err-scale}, it suffices to set $N = \order{\polylog(\mathscr{D}/\epsilon) (\epsilon/\mathscr{D})^{-2}}$, $K = \order{\log(r/\delta)}$, and $L = \order{\log(\epsilon^{-1})}$ to ensure that each individual coefficient can be recovered up to $\abs{\theta_{true}-\theta_{infer}} \leq \epsilon \norm{\Theta}_\infty$ with failure probability at most $\frac{\delta}{r}$. By doing this for each of the $\chi(\mathcal{G}^2) \leq \mathscr{D}^2+1$ partitions and applying a union bound on the failure probability, we see we can recover each coefficient up to an additive error $\epsilon \norm{\Theta}_\infty$ with failure probability at most $\delta$ using $N \cdot L \cdot K \cdot \chi(\mathcal{G}^2)$ queries, which has a complexity given by \cref{eqn:total-query2}. The classical time complexity takes a similar form to the query complexity, except it replaces a factor of $\mathscr{D}^2$ (corresponding to $\chi(\mathcal{G}^2)$) by a factor $r$. This is because we need to process $NLK$ measurement results for each of the $r$ coefficients in the Hamiltonian, whereas for the query complexity, we make $NLK$ measurements for each of the $\chi(\mathcal{G}^2)$ partitions. Since the classical time to color the graph is just $\order{\mathscr{D}^2}$, the overall classical complexity is still dominated by $N \cdot L \cdot K \cdot r$, giving \cref{eqn:total-time2}.
\end{proof}

\section{Hamiltonian Learning with Gibbs States}\label{sec:ham-gibbs}
\begin{lemma}
If $P_i$ is a term in the Hamiltonian, the expectation $\expval{P_i}_\beta$ can also be written as:
\begin{equation}
    \expval{P_i}_\beta = -\frac{1}{\beta} \pdv{\theta_i} \log \Tr \exp(-\beta H)
    \label{eqn:gibbs-expectation-via-derivative}
\end{equation}
\end{lemma}
\begin{proof}
This is Proposition 3.2 of \citet{tang2021}. We reproduce the proof here for completeness:
\begin{align*}
    -\frac{1}{\beta} \pdv{\theta_i} \log \Tr \exp(-\beta H) &= -\frac{1}{\beta} \frac{\Tr \pdv{\theta_i} \exp(-\beta H)}{\Tr \exp(-\beta H)} \\
    &= -\frac{1}{\beta} \frac{\Tr(-\beta P_i \exp(-\beta H))}{\Tr \exp(-\beta H)} \\
    &= \frac{\Tr(P_i \exp(-\beta H))}{\Tr \exp(-\beta H)}.
\end{align*}
\end{proof}
We define $\mathcal{L} = \log \Tr \exp(-\beta H)$ and view this as a function of $\beta, \theta_1,\ldots, \theta_r$.

\begin{lemma}
Using a multivariate Taylor expansion, we can write
\begin{equation}
    \mathcal{L} = \sum_{m \geq 0} \beta^m \sum_{\alpha; \abs{\alpha} = m} \frac{\Theta^\alpha}{\alpha!} \mathcal{D}_\alpha \mathcal{L}, \label{eqn:gibbs-L}
\end{equation}
where $\alpha$ are the multi-index sets defined in \cref{defn:tuples}. The derivative operator $\mathcal{D}_\alpha = \frac{\partial}{\partial z_1^{\alpha_1} \ldots \partial z_r^{\alpha_r}}$ is evaluated at $\Theta=0$. We have defined $z_i \equiv \beta \theta_i$. Furthermore:
\begin{equation}
    \expval{P_i}_\beta = -\sum_{m \geq 0} \beta^{m} \sum_{\alpha; \abs{\alpha}=m} \frac{\Theta^\alpha}{\alpha!} \mathcal{D}_{\alpha'} \mathcal{L} \label{eqn:gibbs-taylor}
\end{equation}
where $\alpha' = (\alpha_1,\ldots,\alpha_i+1,\ldots,\alpha_r)$. 
\end{lemma}
\begin{proof}
The first statement follows directly from Equation (22) of \citet{tang2021}. The second statement follows because the derivatives $\pdv{z_i}$ and $\pdv{z_j}$ commute for all $i,j$, so the operator $-\frac{1}{\beta} \pdv{\theta_i} =-\pdv{z_i}$ from ~\cref{eqn:gibbs-expectation-via-derivative} can be distributed into the sum in \cref{eqn:gibbs-L}.
\end{proof}
\begin{remark}
$\mathcal{D}_\alpha \mathcal{L}$ is a constant that does not depend on $\beta$. To see this, observe that we can write $\mathcal{L}=\log \Tr \exp(-\sum_i z_i P_i)$ as a function of $z_i$ alone. Therefore, evaluating the derivative at $z_1=\ldots=z_r=0$ with the operator $\mathcal{D}_\alpha$ yields a constant independent of $\beta$. So, we are shifting our viewpoint of $\expval{P_i}_\beta$ as a polynomial in $\beta$ with coefficients $\sum_{\alpha; \abs{\alpha}=m} \frac{\Theta^\alpha}{\alpha!} \mathcal{D}_{\alpha'} \mathcal{L}$, rather than as a multivariate polynomial in the coefficients $\Theta$ (as done in \citet{tang2021}).
\end{remark}
\begin{remark}
The zeroth order term in $\expval{P_i}_\beta$ must vanish, since $\expval{P_i}_{\beta=0}$ corresponds to evaluating $P_i$ on the maximally mixed state. 
The first order term is can be found by differentiating \cref{eqn:gibbs-expec} and evaluating at $\beta=0$:
\begin{align*}
    \dv{\beta} \expval{P_i}_{\beta=0} &= \frac{\Tr(\exp(-\beta H))\dv{\beta} \Tr(P_i \exp(-\beta H)) - \Tr(P_i \exp(-\beta H)) \dv{\beta} \Tr(\exp(-\beta H))}{\Tr(\exp(-\beta H))^2} \\
    &= \frac{1}{2^n} \dv{\beta} \Tr(P_i \exp(-\beta H)) \qq{where $n$ is the number of qubits} \\
    &= \frac{1}{2^n} \Tr(-P_i H \exp(-\beta H)) \\
    &= \frac{1}{2^n} \Tr(-P_i H) \\
    &= -\theta_i \qq{by the orthogonality of the Pauli matrices} \numberthis
\end{align*}

\end{remark}

\begin{lemma}[Temperature derivative bound]
The $m$th derivative of $\expval{P_i}_\beta$ evaluated at $\beta=0$ is bounded in absolute value by:
\begin{equation}
    \abs{\dv[m]{\expval{P_i}_\beta}{\beta}}_{\beta=0} \leq \norm{\Theta}_\infty^m (m+1)! (2e^2 (\mathscr{D}^2-1))^{m+1}
\end{equation}
\end{lemma}
\begin{proof}
This follows directly from Lemma 3.7 and Proposition 3.8 of \citet{tang2021}. First, we apply Proposition 3.8 to find that $\abs{\frac{\mathcal{D}_{\alpha'} \mathcal{L}}{\alpha!}} \leq m \abs{\frac{\mathcal{D}_{\alpha'} \mathcal{L}}{\alpha'!}} \leq m (2e(\mathscr{D}+1))^{m+1}$. Next, we observe that $\mathcal{D}_{\alpha'} \mathcal{L}$ is vanishing if $\alpha'$ does not induce a connected subgraph of the Hamiltonian interaction graph. By Lemma 3.7, there are at most $e \mathscr{D} (e (\mathscr{D}-1))^m$ such $\alpha'$. Therefore:
\begin{align*}
    \abs{\dv[m]{\expval{P_i}_\beta}{\beta}}_{\beta=0} &= m! \abs{\sum_{\alpha; \abs{\alpha}=m} \frac{\Theta^\alpha}{\alpha!} \mathcal{D}_{\alpha'} \mathcal{L}} \\
    &\leq \norm{\Theta}_\infty^m(m+1)! e \mathscr{D} (e \qty(\mathscr{D}-1))^m (2e(\mathscr{D}+1))^{m+1} \\
    &\leq \norm{\Theta}_\infty^m (m+1)! (2e^2 (\mathscr{D}^2-1))^{m+1}
\end{align*}
\end{proof}

With unitary dynamics, we already demonstrated that the derivative bound drops by a factor $m!$ if the Hamiltonian is commuting. We expect a similar decrease here, but we leave a proof of this for future works.

\begin{theorem}[Hamiltonian learning with Gibbs states]
The Hamiltonian learning problem (with an additive error $\epsilon \norm{\Theta}_\infty$ and failure probability at most $\delta$) can be solved using
\begin{equation}
    \order{\frac{\mathscr{D}^5 \log (r/\delta) \polylog(\mathscr{D}/\epsilon)}{\epsilon^{2}}} \label{eqn:gibbs-query}
\end{equation}
copies of the Gibbs state. This can be achieved with a time complexity
\begin{equation}
    \order{\frac{\mathscr{D}^4 r \log(1/\delta) \polylog(\mathscr{D}/\epsilon)}{ \epsilon^{2}}}.
\end{equation}
\end{theorem}
\begin{proof}
The protocol is a near mirror image of the Hamiltonian learning protocol using unitary dynamics. We aim to infer the first derivative in the polynomial $\expval{P_i}_\beta$, so we apply our \textsc{EstimateDerivative} protocol from \cref{alg:first-comm}. The error scaling of this protocol is slightly different, since $\norm{f^{(L)}}$ changes using $\expval{P_i}_\beta$ as the polynomial rather than $\expval{P_i(t)}$. Repeating the analysis from \cref{eqn:model-err-bound}, we find:
\begin{align*}
    \abs{c_1 - \mathbb{E}[\tilde{c}_1]} &\leq \frac{L^2 (A/4)^{L-1}}{L!} \abs{\dv[m]{\expval{P_i}_\beta}{\beta}}_{\beta=0} \\
    &\leq \norm{\Theta}_\infty^L L^2 (L+1) (A/4)^{L-1} (2e^2 (\mathscr{D}^2-1))^{L+1} \\
    &= \order{L^3 (A \norm{\Theta}_\infty \mathscr{D}^2)^{L}}
\end{align*}
This amounts to redefining $\gamma = \order{\norm{\Theta}_\infty \mathscr{D}^2}$ (see \cref{defn:typical-scale}). This can be plugged into our analysis for \cref{thm:err-scale} to find that we can infer any individual coefficient up to an error $\epsilon \norm{\Theta}_\infty$ with failure probability $\leq \delta$ using
\begin{equation}
    \order{\log(1/\delta) \polylog(\mathscr{D}/\epsilon) (\epsilon/\mathscr{D}^2)^{-2}}
\end{equation}
copies of the Gibbs state.

This then carries into our main result \cref{thm:master-theorem}. We color our interaction graph such that terms from the same color (i.e., partition) do not have an overlapping support. Then, the coefficients for each partition can be inferred simultaneously because the observables in each partition can be measured simultaneously. Since there are at most $\mathscr{D}+1$ partitions, the overall query complexity of our algorithm is:
\begin{equation*}
    \order{\mathscr{D}^5 \log(1/\delta) \polylog(\mathscr{D}/\epsilon) \epsilon^{-2}}.
\end{equation*}
The processing time is at most
\begin{equation*}
    \order{\mathscr{D}^4 r \log(1/\delta) \polylog(\mathscr{D}/\epsilon) \epsilon^{-2}},
\end{equation*}
since we need to process $\order{\log(1/\delta) \polylog(\mathscr{D}/\epsilon) (\epsilon/\mathscr{D}^2)^{-2}}$ measurements for each of the $r$ coefficients in the Hamiltonian.
\end{proof}

\section{Heuristic Optimizations}\label{sec:heur}
In the following, we provide a theoretical justification for \cref{opt:1,opt:2}. \newline

\noindent \textbf{Optimization 1.}
The support tree $\mathcal{T}_P$ is not often a truly regular tree (i.e., not all nodes have exactly $\mathscr{D}+1$ children). We expect this to be reflected in the scaling of the number of non-vanishing tuples. \cref{cor:counting} says that the number of non-vanishing tuples of size $m$ is $\mathscr{D}^m (m+1)!$. The factor $\mathscr{D}^m$ resembles the number of nodes at a depth $m$ in a $\mathscr{D}$-regular tree. Therefore, it is reasonable to suppose that we can replace $\mathscr{D}^m$ with a factor that more closely reflects the number of nodes at a depth $m$ in the support tree $\mathcal{T}_P$, which is not a truly regular tree. 

We assume $\mathcal{T}_P$ can be modeled as a branching process \citep{athreya1972}. A branching process is a stochastic process that models reproduction over $m$ generations. We begin with a population size of $X_0=1$, and at each time step, each member of the population produces a random number $N$ of offspring, where $N$ is a positive discrete random variable, and the previous generation dies out. Therefore, the distribution of population sizes is
\begin{equation}
    \mathbb{P}(X_{i+1}=x \mid X_i=k)=\mathbb{P}\qty(\sum_{j=1}^k N_j = x).
\end{equation}
It is a standard result that the expected population size at a time $t$ is $\mathbb{E}[N]^t$ \citep{athreya1972}. Applied to our case, we see that a reasonable definition for $N$ is to let it be chosen uniformly at random from the set of degrees in the interaction graph: $N \in_R \qty{\deg(v) \mid v \in V}$ (where $\in_R$ denotes random selection). Furthermore, recall \cref{rem:overcount}: not every subtree of the support tree $\mathcal{T}_P$ is nonvanishing. For any fixed Pauli, the probability that it commutes with some random Pauli is one-half.\footnote{This is because two Pauli operators $A=A_1A_2\ldots A_n$ and $B=B_1 B_2 \ldots B_n$  (with each $A_i, B_i$ being one Pauli matrix) commute when there is an even number of sites where $\qty[A_i, B_i] \neq 0$. Since every single-qubit Pauli matrix does not commute with exactly 2 out of 4 single-qubit Pauli matrices, this yields a one-half probability that $A$ and $B$ commute for $B$ chosen uniformly at random from all $n$-qubit Pauli matrices.} Therefore, as we are constructing a non-vanishing tuple (using the subtree analogy from \cref{cor:counting}), we expect that adding any node has a one-half probability of commuting with the current subtree. Therefore, we define the average degree to be
\begin{equation}
    \overline{\mathscr{D}} = \frac{1}{2\abs{V}}\sum_{v \in V} \deg(v).
\end{equation}
The factor $\frac{1}{2}$ comes from the fact that for a given subtree, we expect one out of every two of the children that we add to commute with the existing operator defined by the subtree. With this, we now replace $\mathscr{D}$ everywhere with $\overline{\mathscr{D}}$.

Numerical simulations indicate the bound tightening in \cref{opt:1} is reasonable (\cref{fig:it-comm}), and significantly improves algorithm performance (\cref{fig:model-err-heur}). Since the theoretical error bound is linear in $\mathscr{D}$, the improvement in performance is linear in $\nicefrac{\mathscr{D}}{\overline{\mathscr{D}}}$, resulting in a decrease in the number of required queries by a factor $\gtrsim \qty(\nicefrac{\mathscr{D}}{\overline{\mathscr{D}}})^2$.
\begin{figure}
    \centering
    \includegraphics{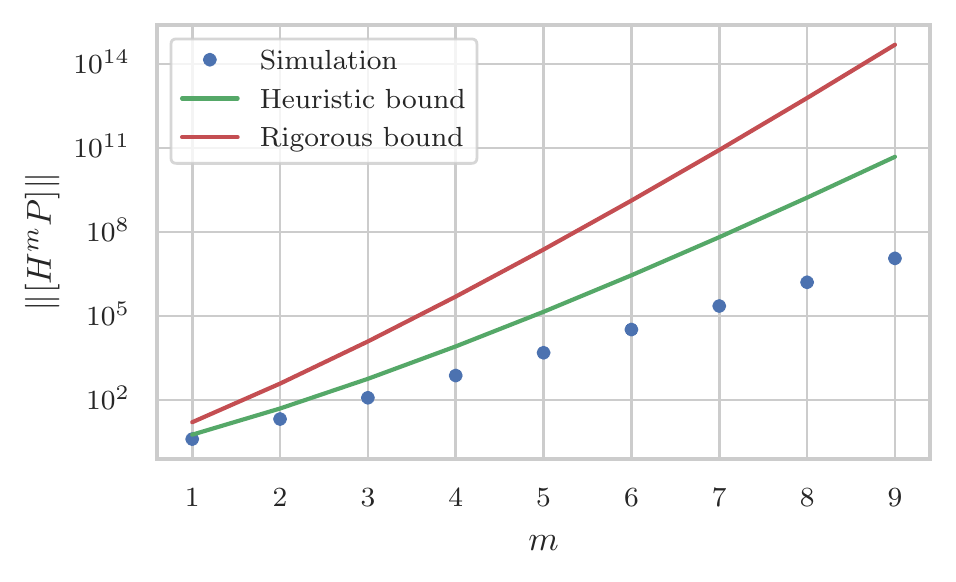}
    \caption[A heuristic improvement for the iterated commutator bound]{Bounds for the iterated commutator norm $\qty[H^m P]$. We randomly generate 10000 TFIM Hamiltonians (sampling the interaction strengths from $\text{Unif}(-1,1)$) for system sizes $n=2,\ldots,9$, and calculate the norm for $\qty[H^m P]$ for every single-qubit Pauli $P$. The points marked simulation in the figure are the maximum $\norm{\qty[H^m P]}$ over all 10000 randomly generated Hamiltonians, all single-qubit Paulis, and all system sizes. The rigorous bound is that found in \cref{cor:counting}: $\mathscr{D}^m (m+1)!$, and the heuristic bound is simply $\overline{\mathscr{D}}^m (m+1)!$. The average degree $\overline{\mathscr{D}}$ was set to $3/2$, which is the effective degree as $n \rightarrow \infty$.}
    \label{fig:it-comm}
\end{figure}

\begin{figure}
    \centering
    \includegraphics{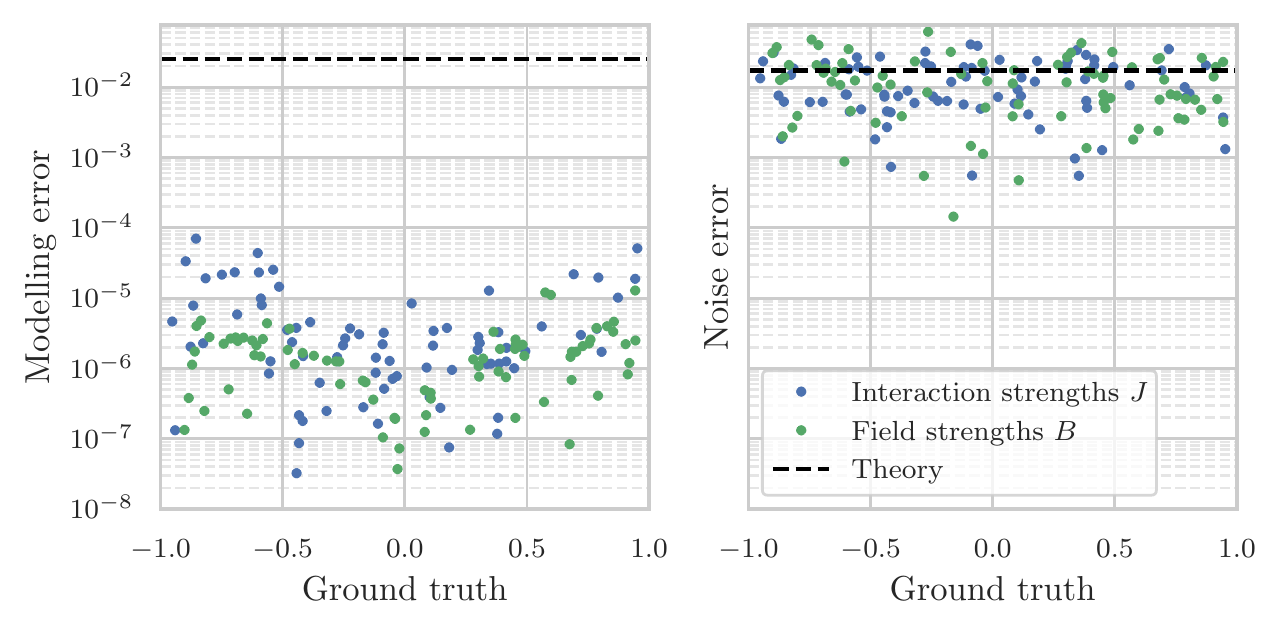}
    \caption[Heuristic improvements to the Hamiltonian learning protocol]{The empirical modelling and noise error of the Hamiltonian learning protocol using the optimal $A$ and $L$. This differs from \cref{fig:model-err} in that here we use $\overline{\mathscr{D}}$ rather than $\mathscr{D}$ when calculating the evolution time. We observe a significantly improved agreement between the theoretical modelling error bound and the true modelling error.}
    \label{fig:model-err-heur}
\end{figure}

\vspace{1em}
\noindent \textbf{Optimization 2.}
For a fixed number of queries, we can optimize the noise error in \cref{eqn:var-est} by allocating a different number of queries to each of the $y_\ell$. If we allocated $N_\ell$ measurements to each time step, we have $\sigma_\ell^2 \leq \frac{1}{N_\ell}$, in which case we solve:
\begin{gather*}
    \underset{\qty{\sigma_\ell}}{\text{minimize}} \qty[\frac{16}{L^2 A^2}\sum_{\ell=1}^{L} \sigma_\ell^2 \qty(\sum_{m=1}^{L-1} (-1)^m m^2 T_m(z_\ell))^2] \\
    \text{subject to: } \sum_{\ell=1}^L N_\ell \leq N_{max}
\end{gather*}
We minimize an upper bound of the objective by replacing $\sigma_\ell^2$ with $\frac{1}{N_\ell}$, in which case the optimization problem can be solved using the method of Lagrange multipliers, and by introducing a slack variable $s$ so that the constraint becomes $\sum_{\ell=1}^L N_\ell - N_{max} - s^2 = 0$. We write $c_\ell \equiv \frac{16}{L^2 A^2} \qty(\sum_{m=1}^{L-1} (-1)^m m^2 T_m(z_\ell))^2$ for brevity and obtain:
\begin{gather*}
    \grad_{\qty{N_\ell}, s, \lambda} \qty[\sum_{\ell=1}^L \frac{c_\ell}{N_\ell} + \lambda \qty(\sum_{\ell=1}^L N_\ell - N_{max} - s^2)] = 0
    \intertext{This reduces to:}
    \frac{N_{\ell_1}}{N_{\ell_2}} = \frac{\sqrt{c_{\ell_1}}}{\sqrt{c_{\ell_2}}} \qc \forall \ell_1, \ell_2 =1,\ldots,L \qq{and} \sum_{\ell=1}^L N_{\ell} = N_{max}
    \intertext{This yields a closed form expression for $N_{\ell}$:}
    N_{\ell} = \frac{\sqrt{c_\ell}}{\sum_\ell \sqrt{c_\ell}} N_{max} \numberthis \label{eqn:opt2}
\end{gather*}

\bibliography{references}

\end{document}